\documentclass[12pt]{article}
\usepackage{array}
\newcolumntype{P}[1]{>{\centering\arraybackslash}p{#1}}
\newcolumntype{M}[1]{>{\centering\arraybackslash}m{#1}}

\usepackage{amsmath,amsthm,amsfonts}
\usepackage{enumerate}
\usepackage{natbib}
\usepackage{url} 

\usepackage{graphicx,psfrag,epsf}

\usepackage{setspace}
 \DeclareGraphicsExtensions{.eps, .ps}
\usepackage{soul,color}

\usepackage{newfloat}
\usepackage{multirow}
\usepackage{arydshln}

\usepackage{xcolor}

\usepackage{xr}
\makeatletter
\newcommand*{\addFileDependency}[1]{
  \typeout{(#1)}
  \@addtofilelist{#1}
  \IfFileExists{#1}{}{\typeout{No file #1.}}
}
\makeatother

\newcommand*{\myexternaldocument}[1]{%
    \externaldocument{#1}%
}
\myexternaldocument{supp}

\def\cond{{\text{cond}}}
\def\bfz{{\bf z}}
\def\EE{\mathbb{E}}

\newtheorem{theorem}{Theorem}[section]
\newtheorem{corollary}{Corollary}[theorem]
\newtheorem{lemma}[theorem]{Lemma}
\newtheorem{assumption}[theorem]{Assumption}
\newtheorem{remark}[theorem]{Remark}

\newtheorem{example}[theorem]{Example}

\renewcommand{\baselinestretch}{0.8}

\newcommand{\blind}{1}

\usepackage{subcaption}

\begin{document}

\def\spacingset#1{\renewcommand{\baselinestretch}%
{#1}\small\normalsize} \spacingset{1}


\if1\blind
{
  \title{\bf Improving prediction in M-estimation by integrating external information from heterogeneous populations}
   \author{
    Walter Dempsey,
    Jeremy M.G. Taylor  \\
    Department of Biostatistics, University of Michigan, USA}
    \date{}
  \maketitle
} \fi

\if0\blind
{
  \bigskip
  \bigskip
  \bigskip
  \begin{center}
    {\LARGE\bf Improving prediction in M-estimation by integrating external information from heterogeneous populations}
\end{center}
  \medskip
} \fi

\bigskip
\begin{abstract}
     A novel approach to improve prediction and inference in M-estimation by integrating external information from heterogeneous populations is proposed. Our method leverages joint asymptotics to combine estimates from external and internal datasets, where the external dataset provides auxiliary information about a subset of parameters of interest. We introduce a shrinkage estimator that combines internal and external estimates under a general class of transformations that ensure consistency across populations. The proposed estimator is shown to achieve improved statistical efficiency compared to using only internal data, with theoretical guarantees on risk reduction. Our approach is particularly valuable in settings where external information is available but populations may differ in their baseline characteristics or effect sizes. We highlight the general applicability by consider applications including generalized linear models, causal inference for treatment effects, missing data, and surrogate endpoint analysis. We demonstrate the method's utility through simulation studies and an analysis of data from the Intern Health Study highlighting its application in time-varying causal moderation analysis with synthetic surrogates.
\end{abstract}



\noindent%
\vfill

\newpage
\spacingset{1.8}

\section{Introduction}

The integration of external information to improve statistical inference has emerged as a critical challenge in modern data science, particularly as researchers seek to leverage auxiliary datasets to enhance the efficiency of their primary analyses. Traditional approaches to this problem have primarily focused on settings where external and internal populations are assumed to be homogeneous, or where external information is treated as population-level summary statistics that constrain internal parameter estimates \citep{10.1214/19-AOAS1293}.  Both generalized method of moments (GMM)~\cite{newey1994large, imbens1994combining} and Empirical Likeihood (EL)~\cite{owen1988empirical, qin1994empirical} approaches have been developed to integrate external information to improve statistical efficiency.   Both GMM and EL consider the external summary information as a fixed, deterministic constraint, a valid assumption only when the external study's sample size is large relative to the internal study.  The Generalized Integration Model (GIM)~\cite{zhang2020generalized} is a likelihood-based data integration technique that improved on the GMM and EL approaches by accounting for the uncertainty inherent in the external summary data.  All of these approaches, however, assume population homogeneity which is often violated in practice where populations may differ in their baseline characteristics, effect sizes, or underlying causal mechanisms. 

More recently, there has been several papers focused on the challenge of combining external and internal data from heterogeneous populations.  Synthetic data integration framework~\cite{gu2023_synthetic} considers stacked imputation with synthetic outcome data to improve statistical efficiency while accounting for population heterogeneity.  A robust fusion-extraction procedure with summary statistics in the presence of biased sources was recently proposed~\cite{wang2023robust} that is asymptotically equivalent to an oracle estimator that uses only unbiased data. 
Alternative shrinkage-based approaches have been developed that allow for combining observational and experimental datasets for causal inference~\citep{10.1111/biom.13827}, and James-Stein type estimators have shown promise in integrating external model information for linear regression settings~\citep{10.1093/biomtc/ujae072}. The existing literature has been limited to allowing for population heterogeneity while making strong distributional assumptions~\cite{10.1093/biomtc/ujae072} or considering specific model classes (e.g., Average Treatment Effect in causal inference). The challenge becomes particularly acute in the M-estimation framework, which encompasses a broad class of statistical methods including generalized linear models, causal inference, missing data problems, and surrogate endpoint analysis.

In this paper, we propose a novel approach that addresses these limitations through a unified framework for integrating external information in M-estimation under heterogeneous populations. Our method leverages joint asymptotics to optimally combine estimates from external and internal datasets, where the external dataset provides auxiliary information about a subset of parameters of interest. The key innovation lies in our introduction of a general class of transformations that ensure consistency across populations, coupled with a shrinkage estimator that achieves improved statistical efficiency compared to using only internal data. Unlike existing methods that rely on constrained optimization or strong homogeneity assumptions, our approach provides theoretical guarantees on risk reduction while accommodating population heterogeneity through flexible transformation functions. Empirically, we see that similar efficiency gains as GIM under population homogeneity, and strong improvements over internal-only estimates in terms of mean squared error under population heterogeneity.

The generality of our framework enables applications across diverse statistical settings. We demonstrate its utility in generalized linear models, where external information about covariate effects can be integrated to improve prediction accuracy. In causal inference contexts, our method allows for the incorporation of external treatment effect estimates while accounting for potential differences in population characteristics. For missing data problems, external information about missingness mechanisms can enhance efficiency, while in surrogate endpoint analysis, external validation studies can inform internal surrogate-to-clinical outcome relationships. Through comprehensive simulation studies and an analysis of data from the Intern Health Study \citep{necamp2020}, we illustrate the method's practical value in time-varying causal moderation analysis.

The remainder of this paper is organized as follows. Section~\ref{sec:shrinktarget} introduces our theoretical framework and establishes the joint asymptotic properties of external and internal estimators. Section~\ref{sec:commonexs} demonstrates the method's application across various statistical settings including linear models, generalized linear models, conditional average treatment effects (CATE), and secondary endpoint analysis. Section~\ref{sec:jsest} presents our proposed shrinkage estimator and derives its theoretical properties. Section~\ref{section:simulations} then provides simulation evidence of its performance.

\section{Shrinkage Target}
\label{sec:shrinktarget}

Let $Y$ be an outcome of interest and $X$ a set of covariates (including an intercept). 
A parameter~$\theta \in \mathbb{R}^p$ is estimated using an external dataset~$D_E = \{ (X_i, Y_i) \}_{i=1}^{n_E}$ by solving a set of estimation equations~$\sum_{i=1}^{n_E} \psi(Y_i, X_i; \theta) = 0$.  We assume the analyst has access to model parameters and associated standard errors, i.e., access $(\hat \theta_{E}, \hat \Sigma_{\theta}^E)$, as well as knowledge of the estimation equations used, $\psi$.  The analyst also has access to individual-level data from an internal dataset~$D_I = \{ (X_i, Z_i, Y_i) \}_{i=1}^{n_I}$ where $Z$ are additional covariates.  A parameter~$\gamma \in \mathbb{R}^q$ is estimated using the internal dataset by solving a set of estimation equations~$\sum_{i=1}^{n_I} \phi(Y_i, X_i, Z_i; \gamma) = 0$.  

The analyst is interested in improving statistical efficiency of the analysis by including the external information. We do not assume correct model specification, i.e., we do not assume~$\psi$ corresponds to score functions for a correctly specified conditional distribution of~$Y$ given $X$.  In contrast to current methods that calibrate the internal solution by solving a constrained optimization problem as in~\cite{doi:10.1080/01621459.2015.1123157} and~\cite{PeisongLawless2019},  our proposed approach is based on joint asymptotics. We start by estimating the same parameter using the internal dataset.  Let~$\hat \theta_I$ denote this estimate.  The proposed method will rely on the following assumption.

\begin{assumption}[Consistency]
\label{consistency}
   The external and internal estimates are asymptotically consistent to the same population parameter under a suitable, differentiable transformation~$h: \mathbb{R}^p \to \mathbb{R}^{p'}$ for $p'\le p$, i.e., $\lim_{n_E \to \infty} h(\hat \theta_E) = \lim_{n_I \to \infty} h(\hat \theta_I) = h^\star$.
\end{assumption}
Assumption~\ref{consistency} states that there exists a known transformation~$h$ of the estimates that will converge to the same population parameter.  The most common example from the literature is the identity transformation $h(\theta) = \theta$ which implies~$\hat \theta_E$ and $\hat \theta_I$ converge to the same limiting population parameter~$\theta^\star$. In this case, a sufficient condition to hold is that the conditional distribution of $Z$ given $(Y,X)$ is equal in law across the internal and external study populations.  Under this assumption, marginalizing over $Z$ in the internal dataset yields equivalent distributions.  The condition is not necessary.  Consider $X$ to be a finite set of $K$ strata, and consider a linear regression that has one parameter per strata, i.e.,~$\theta = (\theta_1, \ldots, \theta_K)$ and $\mathbb{E} [Y|X] = \sum_{k=1}^K \theta_k 1(X = k)$. Then we only require the marginal means per strata to be equal, i.e., $\mathbb{E}_I [ Y | X = k] = \mathbb{E}_E [Y | X = k]$ for all $k=1,\ldots,K$ where the subscript indexes the population. In both settings, use of estimation equations~$\psi$ yield consistent estimates with potentially distinct asymptotic variances.  

A slightly weaker assumption for the conditional distribution~$Y|X$ assumes the causal dynamics may be stable (Penrose, 2004).  In this case, we may have $g_E \{ \mathbb{E}_E (Y | X) \} = \theta_0^E + \sum_{j=1}^p \theta_j^E X_j$ and $g_I \{ \mathbb{E}_I (Y | X) \} = \theta_0^I + \sum_{j=1}^p \theta_j^I X_j$ and the transportability is only in the covariate terms, i.e., $h(\theta) = (\theta_1,\ldots, \theta_p)$ such that $\hat \theta_j^I, \hat \theta_j^E \to \theta_j^\star$ but $\hat \theta_0^I \not = \hat \theta_0^E$. An alternative even weaker assumption is~$\hat \theta_{E,j} = c \hat \theta_{I,j}$for unknown $c \not = 0$ and $j > 1$, i.e., excluding the intercept, reflecting the belief that the relative covariate effects are similar between the two study populations as in~\cite{Taylor2022-xs}.  In this case,~$h(\theta) = (\theta_2/\theta_1,\ldots, \theta_p/\theta_1) \in \mathbb{R}^{p-1}$ satisfies Assumption~\ref{consistency}.

\subsection{An improved estimator based on joint asymptotics}

We next consider joint asymptotics among the external estimate~$\hat \theta_E$ and the two internal estimates $(\hat \theta_I, \hat \gamma)$.  We assume internal and external studies grow  proportionally to one another, i.e., $n = n_I = c n_E$ for $c \in (0,\infty)$. Here, the constant $c$ controls the relative rate of growth between the two studies. When $c \to 0$ corresponds to the setting where the external study is considered a population-level study and our approach will be shown to mirror prior approaches that calibrate internal analysis to external population-level summary statistics.  

\begin{lemma}[Joint asymptotics]
\label{lemma:jointasymp}
Let~$\nabla h(\theta^\star_E) = \nabla h_E$ and $\nabla h(\theta^\star_I) = \nabla h_I$ denote the gradient of $h$ evaluated at the external and internal population parameters.  
Then we have $\sqrt{n} \left( 
    h(\hat \theta_E) - h (\theta^\star_E),
    h(\hat \theta_I) - h (\theta^\star_I),
    \hat \gamma_I - \gamma^\star 
    \right)
    $ converges in distribution to a mean-zero Normal distribution with covariance
    \begin{equation}
    \label{eq:jointasymptotics}
    \left( 
    \begin{array}{ccc}
         c^{-1} \nabla h_E^\top \Sigma^E_{\theta^\star} \nabla h_E & {\bf 0}_{p} & {\bf 0}_{q}   \\
        {\bf 0}_{p} &  \nabla h_I^\top \Sigma^I_{\theta^\star} \nabla h_I &  \, \Sigma^I_{\theta^\star, \gamma^\star} \nabla h_I \\
        {\bf 0}_{q} & \nabla h_I^\top \Sigma^I_{\theta^\star, \gamma^\star} & \Sigma^I_{\gamma^\star}
    \end{array}
    \right),
\end{equation}
where~$\Sigma_{\theta^\star}^E$, $\Sigma_{\theta^\star}^I$ and $\Sigma_{\gamma^\star}^I$ are the asymptotic variances associated with estimating equations for $\theta$ in the external and internal studies and $\gamma$ for the internal study respectively, and $\Sigma_{\theta^\star, \gamma^\star}$ is the covariance between the estimating equations. 
\end{lemma}
Lemma~\ref{lemma:jointasymp} is a direct consequence of standard Z-estimation theory and application of the Delta-method~\citep{Vaart_1998}. An estimate of the asymptotic variance-covariance matrix for the external parameter is provided.  Asymptotic variance-covariance matrix~$\Sigma^I$ for the internal parameters~$(\hat \theta_I, \hat \gamma_I)$ is equal to the sandwich covariance~$Q^{-1} W Q^{-1}$ where
\begin{equation*}
Q = \mathbb{E}_I \left[ 
\begin{array}{c c} 
\frac{\partial}{\partial \theta} \psi(Y,X; \theta) & {\bf 0}_{p \times q} \\
{\bf 0}_{q \times p} & \frac{\partial}{\partial \gamma} \phi(Y,X,Z; \gamma) \\
\end{array}
\right], \quad
W = 
\mathbb{E}_I \left[ 
 \left( \begin{array}{c}
      \psi (Y,X;\theta) \\
      \phi(Y,X,Z;\gamma) 
 \end{array}
 \right) 
 \left( \begin{array}{c}
      \psi (Y,X;\theta) \\
      \phi(Y,X,Z;\gamma) 
 \end{array}
 \right)^\top 
 \right].
\end{equation*}
Both $Q$ and $W$ can be estimated empirically by empirical averages using the internal study.  

Under Assumption~\ref{consistency}, equation~\eqref{eq:jointasymptotics} implies the following joint distribution 
  \begin{equation}
    \label{eq:jointasymptotics2}
    \sqrt{n} \left(  
    \begin{array}{c}
    h(\hat \theta_E) - h(\hat \theta_I) \\
    \hat \gamma_I - \gamma^\star 
    \end{array}
    \right) \to N \left( {\bf 0}_{p + q}, 
    \left( 
    \begin{array}{cc}
         c^{-1} \nabla h_E^\top \Sigma^E_{\theta^\star} \nabla h_E + \nabla h_I^\top \Sigma^I_{\theta_\star} \nabla h_I & \Sigma_{\theta^\star, \gamma^\star}^I \nabla h_I \\
        \nabla h_I^\top \Sigma_{\gamma^\star, \theta^\star}^I & \Sigma_{\gamma^\star}^I
    \end{array}
    \right)
    \right).
\end{equation}
To reduce notation complexity, let $\Sigma^h_{\theta^\star} := c^{-1} \nabla h_E^\top  \Sigma^E_{\theta^\star} \nabla h_E + \nabla h_I^\top \Sigma^I_{\theta_\star} \nabla h_I$ and $\Sigma_{\theta^\star, \gamma^\star}^h = \Sigma_{\theta^\star,\gamma^\star}^I \nabla h_I$.  Then the joint distribution above implies the conditional distribution for $\sqrt{n} (\hat \gamma_I - \gamma^\star)$ given the difference between the transformed external and internal estimates is
$$
\sqrt{n} (\hat \gamma_I - \gamma^\star) | \sqrt{n} (h(\hat \theta_I) - h(\hat \theta_E))
\to N \left( \Sigma^h_{\gamma, \theta} \left( \Sigma^h_\theta\right)^{-1} \sqrt{n} (h(\hat \theta_I) - h(\hat \theta_E)), \Sigma_{\gamma^\star}^I - \Sigma^h_{\gamma, \theta} \left( \Sigma^h_{\theta} \right)^{-1} \Sigma^h_{\theta, \gamma} \right).
$$
The conditional distribution suggests an improved, unbiased estimator of~$\gamma$ can be given by 
\begin{equation}
\label{eq:newestimator}
\hat \gamma_{\cond} = \hat \gamma_I - \Sigma^h_{\gamma^\star, \theta^\star} \left( \Sigma^h_{\theta^\star}\right)^{-1} (h(\hat \theta_I) - h(\hat \theta_E)).
\end{equation}
We refer to~\eqref{eq:newestimator} as the {\it conditional estimate} as it leverages the asymptotic conditional distribution to improve efficiency.  Corollary~\ref{corollary:efficiency} demonstrates this claim to be true asymptotically.

\begin{corollary}
\label{corollary:efficiency}
    Under Assumption~\ref{consistency}, the conditional estimate given by~\eqref{eq:newestimator} is asymptotically more efficient than the internal-only estimate. 
\end{corollary} 
The claim is true by examination of the quadratic forms, i.e., $x^\top \left( \Sigma_{\gamma^\star} - \Sigma^h_{\gamma^\star, \theta^\star} \left( \Sigma^h_{\theta^\star}\right)^{-1} \Sigma^h_{\theta^\star, \gamma^\star} \right) x \leq x^\top \Sigma_{\gamma^\star} x$ for any $x \in \mathbb{R}^q$ since $\Sigma^h_{\gamma^\star, \theta^\star} \left( \Sigma^h_{\theta^\star} \right)^{-1} \Sigma^h_{\theta^\star, \gamma^\star}$ is positive semi-definite.  

Prior work on constrained maximum likelihood estimation guarantees efficiency gains as a restricted maximum likelihood estimate is guaranteed to be more efficient than an unrestricted one if the restriction is correct. Here, we do not perform likelihood-based inference and therefore rely on joint asymptotics from Z-estimation theory.


\begin{remark}[Generalized Method of Moments]
Our framework can be extended to allow the analyst to use generalized method of moments (GMM), where estimators are defined through moment conditions of the form $\mathbb{E}[\psi(Y,X;\theta)] = 0$ for some vector-valued function $\psi$. The GMM estimator $\hat{\theta}$ solves $\sum_{i=1}^{n} \psi(Y_i, X_i; \theta) = 0$ and satisfies $\sqrt{n}(\hat{\theta} - \theta_0) \to_d N(0, \Sigma_\theta)$ where $\Sigma_\theta = (G^\top W G)^{-1} G^\top W \Omega W G (G^\top W G)^{-1}$, with $G = \mathbb{E}[\partial \psi(Y,X;\theta_0)/\partial \theta]$, $\Omega = \mathbb{E}[\psi(Y,X;\theta_0)\psi(Y,X;\theta_0)^\top]$, and $W$ is a weighting matrix \citep{hansen1982}. Since GMM estimators rely on asymptotic normality and our joint asymptotic framework in Section~\ref{sec:shrinktarget} requires only this property, our proposed conditional and James-Stein estimators can be directly applied to GMM settings without modification. 
\end{remark}

\begin{remark}[Multiple external studies]
A common setting is the existence of multiple external studies that may provide information about~$\theta$. If each external study satisfies Assumption~\ref{consistency}, then our proposed approach naturally extends to handle multiple external studies; however, it is often the case that a subset of external sources are biased due to issues like biased sampling, data corruption, or model misspecification.  Wang et al. (2023)~\cite{wang2023robust} proposed a robust fusion-extraction procedure which yields a consistent estimator~$\tilde \theta$ from multiple studies.  A key condition is that more than half of the data come from unbiased data sources.  In this case, we can apply their procedure first to build an external estimate and then apply our proposed method to the resulting estimate. 
\end{remark}

\subsection{Some common examples}
\label{sec:commonexs}

We start by showing how our proposed approach can be applied in three common settings, highlighting two useful theoretical results in the process.

\begin{example}[Linear Models]
\label{ex:linear}
\normalfont
Consider the linear model $E[ Y | X,Z] = X^\top \gamma_x + Z^\top \gamma_z$. Assume that the external set of estimating equations solve $\sum_{i=1}^{n_e} X_i (Y_i - X_i^\top \theta)  = 0$.  Assume the same set is used with the internal study and the transformation to satisfy Assumption~\ref{consistency} is the identity. 

For each element of $Z$, first solve the following set of estimation equations $\sum_{i=1}^{n_I} X_i (Z_j - X_i^\top \beta_j) = 0$. Let~$\tilde Z_i = Z_i - \hat B X_i$ where $\hat B \in \mathbb{R}^{q \times p}$ has the $j$th row equal to $\hat \beta_j$.  Now consider the set of estimating equations for the internal data given by

$$
\sum_{i=1}^{n_I} \left( \begin{array}{c} X_i \\ \tilde Z_i \end{array} \right) (Y_i - X_i^\top \gamma_x - \tilde Z_i^\top \gamma_z) = {\bf 0}_{p+q}
$$
This is a linear transformation of the original model where $\tilde Z$ replaces the set of additional covariates~$Z$, made orthogonal to $X$ so the estimating equations satisfy
\begin{equation}
     \label{eq:centeredls}
\sum_{i=1}^{n_I} \left[ X_i (Y_i - X_i^\top \tilde \gamma_x - \tilde Z_i^\top \tilde \gamma_z) \right] = 
\sum_{i=1}^{n_I} \left[ X_i (Y_i - X_i^\top \tilde \gamma_x )\right] = 0.
\end{equation}
This implies~$\hat {\tilde \gamma}_x = \hat \theta_I$ and $\gamma_z$ is asymptotically orthogonal to the internal study $\theta$-estimates.  Returning to the improved estimator, this implies~$\Sigma^h_{\gamma_x, \theta} = \Sigma_{\theta, \theta}$ and $\Sigma_{\gamma_z, \theta} = 0$. Thus
\begin{equation*}
\begin{split}
\hat \gamma_{\text{cond}} &= \left( 
\begin{array}{c}
     \hat \gamma_x \\ 
     \hat \gamma_z 
\end{array}
\right) 
- 
\left( \begin{array}{c}
     \Sigma_{\theta^\star}^I (c^{-1} \Sigma_{\theta^\star}^E + \Sigma_{\theta^\star}^I)^{-1} \\
     {\bf 0}_{q \times p}
\end{array}
\right) ( \hat \theta_I - \hat \theta_E) \\
&= 
 \left( 
\begin{array}{c}
     \left( (\Sigma_\theta^I)^{-1} + c (\Sigma_\theta^E)^{-1} \right)^{-1}  (\Sigma_\theta^I)^{-1} \hat \theta_I + 
     \left( (\Sigma_\theta^I)^{-1} + c (\Sigma_\theta^E)^{-1} \right)^{-1} 
     c (\Sigma_\theta^E)^{-1} \hat \theta_E  \\ 
     \hat \gamma_z 
\end{array}
\right).
\end{split}
\end{equation*}
where the second equality is by direct application of the Woodbury identity.
This shows the estimate of $\hat \gamma_x$ is a precision weighted average of the internal and external estimates, while $\hat \gamma_z$ does not change as it is asymptotically orthogonal, implying external information about $\theta$ provides no information for improving our estimate.  
If~$\Sigma_\theta^I = \Sigma_\theta^E$ then the weights become $(1+c)^{-1}$ and $c(1+c)^{-1}$ which is equivalent to relative sample size weighting. 


     

\end{example}

Example~\ref{ex:linear} shows that the proposed conditional method leads to reasonable behavior in the linear setting.  Moreover, it establishes a connection with the constrained maximum likelihood literature.  Our next example considers an extension to generalized linear models.  We establish reasonable behavior in this setting and discuss distinctions from the constrained maximum likelihood approach.

\begin{example}[Generalized Linear Models]
\label{ex:glms}
\normalfont
Consider the generalized linear model $g(\EE[Y|X,Z]) = X^\top \gamma_x + Z^\top \gamma_z$ where $g$ is the link function.  In most settings, the non-linear link function makes such a model non-collapsible.  That is, if the true model is given by the above, then $g(E[Y|X]) = X^\top \theta$ cannot be true model. Theorem~\ref{thm:glms} demonstrates reasonable behavior of our conditional estimator in this setting.

\begin{theorem}
\label{thm:glms}
    Let $(Y_i, X_i,Z_i)$ denote independent and identically distributed random variables sampled from a joint probability function~$\mathcal{P}$, with $Y_i$ an outcome from a generalized linear model with canonical link $g$ with covariates $(X_i,Z_i)$ as covariates. Consider two generalized linear models:
    \begin{equation*}
        \begin{split}
            g(\EE[Y | X]) &= \theta_0 + X^\top \theta_X \\
            g(\EE[Y | X, Z]) &= \gamma_0 + X^\top \gamma_X + Z^\top \gamma_Z \\
        \end{split}
    \end{equation*}
    Assume~$\psi(Y, X; \theta)$ and $\phi(Y, X, Z; \gamma)$ are the score equations for maximum likelihood estimation under the above potentially misspecified models.  Then asymptotically the $Z$-component of $\hat \gamma_{\text{cond}}$ will converge to the  $\hat \gamma_{Z}^{(I)}$, i.e., there is no asymptotic improvement to the estimated coefficients of $Z$.  
\end{theorem}

Theorem~\ref{thm:glms} implies that if the external study had estimated a misspecified logistic regression on $X$, then our conditional estimate of parameters related to $Z$ cannot be improved by conditioning on the estimated difference~$\theta_I - \theta_E$.  This behavior has been observed empirically for the constrained maximum likelihood~\citep{gu2023_synthetic,gu2021_metainference}.  Theorem~\ref{thm:glms} will be empirically demonstrated in Section~\ref{section:simulations}.



\end{example}

\begin{example}[CATE]
\label{ex:cate}
\normalfont
Consider heterogeneous effect estimation which plays a crucial role in causal inference with applications across medicine and social science. The most common target parameter in this setup is the conditional average treatment effect (CATE) function. We consider $n$ i.i.d observations $\{X_{i}, Z_i, A_{i}, Y_{i} \}_{i=1}^{n_I}$ where $A_i \in \{0,1\}$ is a binary treatment, and the data could arise from either a randomized controlled trial or an observational study.  Using the potential outcomes framework~\citep{potentialoutcome}, let $Y_{i}(a)$ denote the counterfactual outcome if the treatment is set to $a \in \{0,1\}$.   Then the conditional average treatment effect (CATE) is defined as $\Delta(x,z) =\mathbb{E} \left[Y(1) - Y(0)|X=x, Z = z\right]$.  There has been a growing interest in data integration to estimate average treatment effects~\citep{RosenmanPropensity2022, Han21012025}. There is also a growing related literature on  generalizability and transportability of results from randomized trials to target populations of interest. See~\cite{Degtiar2023} for a comprehensive review. Here we aim to leverage an external observational or experimental dataset to provide unbiased, efficient, and robust estimation of conditional average treatment effects for the internal study population. A key distinction from the prior work is the focus on the CATE given $(X,Z)$ for the internal study population where $Z$ is unobserved in the external study.

Standard assumptions of consistency, ignorability, and positivity are made for the internal study~\citep{doi:10.1080/03610929108830654, 10.1007/978-1-4612-1842-5_4}.  Specifically, ignorability is assumed given the pair $(X,Z)$.  We consider a parametric model for the CATE, i.e.,~$\Delta(x,z) = f(x,z)^\top \gamma$ where $f(x,z) \in \mathbb{R}^p$ is a set of constructed features.  When model misspecification occurs, we can still interpret the proposed linear form as an $L_2$ projection of the true causal effect onto the space spanned by a feature vector $f(x,z)$ that only depends on x and z~\cite{10.1214/19-AOAS1293}. The choice between these interpretations reflects a bias-variance trade-off. In practical applications, the projection interpretation ensures a well-defined parameter with practical interest.

For simplicity, we assume known propensities denoted by $p_I(x,z)=P_I(A = 1 | X=x,Z=z)$. Then a fully parametric approach for estimating~$\gamma$ is given by the following set of estimating equations
\begin{equation}
    \label{eq:wcls}
\sum_{i=1}^{n_I} \left(Y_i - g(X_i, Z_i)^\top \alpha - (A_i - p_I(X_i,Z_i)) f(X_i, Z_i)^\top \gamma \right)
\left(
\begin{array}{c}
     g(X_i, Z_i) \\
     (A_i - p_I(X_i,Z_i)) f(X_i,Z_i)
\end{array}
\right)
= 0
\end{equation}
where $\alpha$ is a nuisance parameter and $g(x,z)^\top \alpha$ is a working model for $\mathbb{E} [ Y | X=x,Z=z]$. It is common to assume the feature vector~$g(x,z)$ contains~$f(x,z)$~\cite{10.1214/19-AOAS1293}.  Equation~\eqref{eq:wcls} corresponds to the $\phi$ equation for the internal data in our proposed framework. For the external study, we assume~$\{X_i, A_i, Y_i\}_{i=1}^{n_E}$ are i.i.d. and consider leveraging~\eqref{eq:wcls} with $g(X)^\top \tilde \alpha$, $f(X)^\top \theta$, and the propensity $p_E (X)$ replacing their respective terms above, to give the $\psi$ estimating equations for the external study in our proposed framework.  
When the data arise from observational studies, the propensity models will not be known.  In these settings, common approaches such as sample splitting~\citep{chernozhukov2018double,semenova2021, chernozhukov2017} can be applied to estimate both nuisance functions (i.e., propensity and working models for $\mathbb{E}[Y|X=x,Z=z]$ and $\mathbb{E}[Y|X=x]$ in $\psi$ and $\phi$ respectively). 

Assumption~\ref{consistency} will hold based on the above~$\phi$ and $\psi$ if standard causal assumptions of consistency, ignorability, and positivity hold and the conditional average treatment effect given $X$ is transportable across the internal and external studies, i.e., $\Delta_I (x) = \Delta_E (x)$.  If this does not hold, model- and weighting-based methods do exist that make Assumption~\ref{consistency} hold under slightly different choice of estimating equations. See \cite{Han21012025} and \cite{huch2024dataintegrationmethodsmicrorandomized} for two recent papers that demonstrate how to alter the estimating equations to robustly identify ATEs and parametric CATEs respectively.  Interestingly, Assumption~\ref{consistency} can hold even if the standard causal assumptions do not hold.  Consider the setting where~$Z$ is an important confounder and the internal and external study data are independent and identically distributed.  Then Assumption~\ref{consistency} holds since both~$\hat \theta_E$ and $\hat \theta_I$ converge to~$\theta^\star$; however,~$\theta^\star$ has no causal interpretation.  In this case, we are simply leveraging the correlation between the difference in two non-causal parameters, i.e.,~$\hat \theta_E - \hat \theta_I$, to improve efficiency of the causal parameter~$\gamma$.

\begin{remark}[Control-only external study]
Example~\ref{ex:cate} considers the setting where the external study includes treatment.  Alternatively, control condition only external data may be available. That is the external study is simply~$\{ X_i, 0, Y_i \}$ where $A_i = 0$ for all individuals in the external study.  Using~\eqref{eq:wcls} may not improve efficiency in this case.  However, an alternative set of estimating equations for~$\psi$ can be specified as 
$$
\sum_{i=1}^{n_E} (Y_i - A_i g(X_i)^\top \theta_0 - (1-A_i) g(X_i)^\top \theta_1)
\left(
\begin{array}{c}
     A_i g(X_i) \\
     (1-A_i) g(X_i)
\end{array}
\right)
= 0,
$$
Let~$\theta = (\theta_0, \theta_1)$ and $h(\theta) = \theta_0$.  Then the above equation would correspond to the estimating equations $\psi$ in our proposed framework.  The equations $\phi$ allow for pooling information from the external study through~$\theta_0$, i.e., the conditional mean model for the control arm. We can continue to use~\eqref{eq:wcls} to estimate the conditional average treatment effect, i.e., it still corresponds to the estimating equations for $\phi$ in our proposed framework.
\end{remark}

\begin{remark}[Conditional Average Relative risks]
The linear contrast is a natural candidate when considering continuous outcomes; however, if~$Y$ is binary, then an alternative to the CATE is the conditional log-relative risk given by 
$$\Delta(x,z) = 
\log \left( \frac{\mathbb{E} \left[Y(1)|X=x, Z=z\right]}{\mathbb{E} \left[Y(0)|X=x, Z=z\right]} \right) = \log \left( \frac{\mathbb{E} \left[Y| A = 1, X=x, Z=z\right]}{\mathbb{E} \left[Y| A = 0, X=x, Z=z\right]} \right).$$  The second equality follows under standard causal assumptions of consistency, ignorability, and positivity.  The conditional log-relative risk can be estimated via the estimating equations
$$
\sum_{i=1}^{n_I} e^{-A_i f(X_i,Z_i)^\top \gamma} \left(Y_i - e^{g(X_i,Z_i)^\top \alpha + A_i f(X_i,Z_i)^\top \gamma} \right)
\left(
\begin{array}{c}
     g(X_i,Z_i) \\
     (A_i - p_I(1|X_i)) f(X_i,Z_i) 
\end{array}
\right) =0.
$$    
The above estimating equations correspond to $\phi$ in our proposed framework.  The estimating equations $\psi$ would be given by the above equation with $g(X_i)^\top \tilde \alpha$ replacing $g(X_i,Z_i)^\top \alpha$ and $f(X_i)^\top \theta$ replacing $f(X_i,Z_i)^\top \gamma$. 
Note that there is no equivalent generalized linear model formulation of the above approach.  This implies no equivalent GIM-based method (even under stronger distributional assumptions).
\end{remark}
\end{example}


\subsection{Examples involving secondary endpoints}
\label{sec:secendpts}

The three previous examples are standard data integration problems in which the observed outcome~$Y$ is the same in the external and internal studies, with the only difference being the auxiliary covariate~$Z$ in the internal study.  
We next consider three more examples that involve a secondary endpoint to demonstrate the broad applicability of our proposal beyond the traditional data integration setting. 

\begin{example}[Secondary endpoints]
\label{ex:secendpts}
\normalfont

In practice, the outcome measured in the external study may not be the outcome of interest in the internal study.  Consider the scenario where the internal study measures $(Y_1, Y_2, X,Z)$ and our primary outcome of interest is $Y_1$ while the external study only measures $(Y_2, X)$. We refer to $Y_2$ as the secondary endpoint in the internal study.  Constrained methods such as~\cite{10.1093/biomtc/ujae072} and \cite{doi:10.1080/01621459.2015.1123157} are not designed for this scenario.  Our proposal, on the other hand, can be applied directly. 

To demonstrate the utility of our proposed method for secondary endpoints, we consider the simple setting where there is no additional covariate~$Z$, and $(Y_1,Y_2) | X$ is multivariate normal with mean that is linear in $X$ (i.e., $\mu = (X^\top \gamma, X^\top \theta)$ and variances~$\sigma_j^2$ and covariance~$\sigma_1 \sigma_2 \rho$. Finally, we assume that the marginal distribution of $X$ is the same in the internal and external studies.  Under these assumptions,~$\Sigma_{\gamma, \gamma} = \sigma_1^2 \EE(X X^\top)^{-1}$ and $\Sigma_{\gamma, \theta} = \sigma_1 \sigma_2 \rho \EE(X X^\top)^{-1}$.  Therefore, the conditional estimator can be expressed simply as
\begin{equation}
    \label{eq:condapproach_secendpt}
\gamma_{\text{cond}} = \hat \gamma^{(I)} +
\frac{c}{1+c} \rho \frac{\sigma_1}{\sigma_2} \left(\hat \theta^{(I)} - \hat \theta^{(E)} \right).
\end{equation}
That is, the conditional estimator takes the internal only estimator~$\hat \gamma^{(I)}$ and scales the error in the secondary endpoint parameter estimates between internal and external studies~$\left(\hat \theta^{(I)} - \hat \theta^{(E)} \right)$ by the correlation~$\rho$ that has been properly rescaled by the ratio of variances $\sigma_1/\sigma_2$ and the relative sample sizes~$c/(1+c)$.   
\end{example}

\begin{example}[Predictive-based secondary endpoints]
\label{ex:predsecendpt}
\normalfont
Consider the scenario where the internal study did not include~$Y_2$ as a secondary endpoint. To handle this, we require additional information from the external study.  First, assume that the external data is split into two independent subsets.  The first subset is used to fit a predictive model of the outcome~$Y_2$ given covariate~$X$.  Let~$\tilde Y = f(X)$ denote this prediction.  The second subset is then used to solve the estimating equations~$\sum_{i=1}^{n_i} \psi(\tilde Y_i, X_i; \theta)$.  Assume that the analyst has access to the predictive model as well as the point estimate~$\hat \theta_E$ and variance-covariance~$\Sigma_{\hat \theta}^E$.  The data analyst then generates a secondary endpoint~$\tilde Y$ using the predictive model and performs inference in the same way as example~\ref{ex:secendpts} above.  The use of the predictive model significantly expands the scope of secondary endpoints as it allows data integration even when the internal study did not have direct access to the secondary endpoint as long as the internal study measures the necessary covariates to perform prediction.
\end{example}

\begin{example}[Missing data as a data integration problem]
\label{ex:missingdata}
\normalfont
We next discuss how our proposed framework can handle two important missing data scenarios.

In the first scenario, the data are given by $\{ X_i, R_i, R_i Y_i\}$ where $R_i$ is a missing data indicator. We assume data are missing-at-random (MAR) given the covariate $X$ and we are interested in fitting the regression model~$\psi(Y, X; \theta)$. Under correct model specification, a complete case-analysis will not be biased; however, it will be inefficient. A common approach is to consider multiple imputation. Here, we provide an alternative when we have access to an external predictive model as in Example~\ref{ex:predsecendpt} above. For simplicity, we assume no additional external information. As above, we construct a secondary endpoint~$\tilde Y$ using the external predictive model. We split the internal study into two datasets depending on whether the primary outcome is observed. The subset where the primary outcome is not observed is treated as the external study, i.e., $\{ (X_i, R_i=0, \tilde Y_i) \}$, while the subset where the primary outcome is observed is treated as the internal study, i.e., $\{ (X_i, R_i = 1, Y_i, \tilde Y_i) \}$. Under the missing-at-random assumption, Assumption~\ref{consistency} holds and the setting is identical to Example~\ref{ex:secendpts} where~$\tilde Y$ is a secondary endpoint. The asymptotic efficiency gain depends on the level of correlation between the primary outcome and the derived secondary outcome. Under the normality assumptions in Example~\ref{ex:secendpts}, if the data are MCAR, then the conditional estimator takes the same form as~\eqref{eq:condapproach_secendpt} with the term~$c/(1+c)$ being the relative amount of missing data.

In the second scenario, the data are given by $\{ X_i, R_i, R_i Z_i, Y_i\}$ where $R_i$ is a missing data indicator for the additional covariate $Z$. In this scenario, we assume the data are missing-at-random given the joint covariates $X$ and $Z$.  Then a complete case analysis using the subset of the data where both $X$ and $Z$ are observed is consistent. To improve efficiency, the analysis can split the data based on availability of the additional covariate. By taking the subset where~$Z$ is unobserved as an external study and the subset where $Z$ is observed as our internal study, Assumption~\ref{consistency} will hold and, therefore, the conditional approach can lead to efficiency gains even though the missing-at-random assumption does not hold for the external study.  
\end{example}

\section{James-Stein shrinkage estimator}
\label{sec:jsest}

Section~\ref{sec:shrinktarget} provides efficient estimation of the parameter of inference under Assumption~\ref{consistency}. When asymptotic consistency across internal and external studies is not satisfied, we may still consider~$\hat \gamma_{\text{cond}}$ from equation~\eqref{eq:newestimator} as a shrinkage target and perform data adaptive shrinkage towards the estimate. We propose a James-Stein shrinkage estimator that seeks to minimize a weighted quadratic loss~$L(\hat \gamma, \gamma) = ( \hat \gamma - \gamma)^\top A ( \hat \gamma - \gamma)$, where~$A \in \mathbb{R}^{q \times q}$ is a known, weight matrix.  Common choices of $A$ include (1) the identity to consider the mean squared error (2) the inverse of the variance-covariance matrix to consider the standardized mean square error, (3) the expected value of the design matrix to consider predictive mean square error, i.e., let $H$ be the design matrix associated with the estimating equations~$\psi$ then $A = \EE_I(H H^\top)$, and (4) a subset of the design matrix to consider predictive mean square error on a subset of the covariates (e.g., predictive CATE mean square error).  Given this loss, we propose a James-Stein shrinkage estimator
$$
\hat \gamma_{JS} = \hat w \hat \gamma_{\text{cond}} + (1 - \hat w) \hat \gamma_I
$$
with weight
\begin{equation}
\label{eq:jsweight}
\hat w = \left( 1 - \frac{\hat \tau}{n \left( \hat \gamma_{\cond} - \hat \gamma_I \right) A \left( \hat \gamma_{\cond} - \hat \gamma_I \right) } \right)_{+}    
\end{equation}
where~$(\bullet)_+$ denotes the positive part of the argument, and $\tau$ is a parameter controlling the amount of shrinkage.  The weight $\hat w$ sits between 0 and 1, with large weights indicating similarity between the conditional and internal estimates in terms of the weighted quadratic loss and thus the James-Stein estimator being set equal to the conditional estimate.  When the weight is near 0, the opposite is true (i.e., a high degree of dissimilarity) and the James-Stein estimator is set equal to the internal-only estimate.  James-Stein shrinkage estimators using external study information have been previously proposed in the literature.  Shrinkage estimators have been proposed for average treatment effects (ATEs)~\cite{10.1111/biom.13827} using the external ATEs as the shrinkage targets, as well as for linear models~\cite{10.1093/biomtc/ujae072} using constrained MLEs as the shrinkage targets.  Both focus on quadratic loss and demonstrate a risk reduction of the James-Stein estimator.  Here, we consider our shrinkage target derived from the joint asymptotics which is applicable across the wide array of problem settings described in Sections~\ref{sec:commonexs} and~\ref{sec:secendpts}. 

\begin{remark}[Relationship to differences in $\theta$]
    Under the identify transformation~$h(\theta) = \theta$, equation~\eqref{eq:newestimator} implies the difference~$\hat \gamma_{\cond} - \hat \gamma_I$ is a scaled-version of the difference between external and internal estimates, $\hat \theta_E - \hat \theta_I$. This implies that the  weights can be re-written as
    $$
    \hat w = \left( 1 - \frac{\hat \tau}{n \left( \hat \theta_E - \hat \theta_I \right) \Sigma_{\theta}^{-1} \Sigma_{\theta, \gamma} A \Sigma_{\gamma, \theta} \Sigma_{\theta}^{-1} \left( \hat \theta_E - \hat \theta_I \right)} \right)_{+}.
    $$
    This clarifies how the external data is being used within our framework.  The weighted quadratic loss is being translated to the $\theta$-scale and large-scale deviations from Assumption~\ref{consistency} will imply $\hat w \to 0$ and thus $\hat \gamma_{JS} = \hat \gamma_I$.  Under general transformations~$h$, the same intuition holds with the weighted quadratic loss being translated to the $h(\theta)$-scale. 
\end{remark}

When Assumption~\ref{consistency} does not hold, study population heterogeneity implies the JS-type estimator will not be consistent for $\gamma^\star$. To compare estimators, we closely follow~\cite{HANSEN2016115} in calculating asymptotic risks defined by
\begin{equation}
\label{eq:asymprisk}
    R(\hat \gamma, \gamma) = \lim_{\zeta \to \infty} \lim_{n \to \infty} \mathbb{E}_{I} \left [\min( n L(\hat \gamma, \gamma), \zeta) \right ].
\end{equation}
The expectation is over the internal distribution, the scaled loss trimmed at value~$\zeta$ where this term is asymptotically neglibile.  Recall the loss we consider is the weighted quadratic loss, but we define asymptotic risk generally.  

If Assumption~\ref{consistency} does not hold, then in many settings the estimator~$\hat \gamma_{JS}$ will not be consistent and~$\hat w \to^p 0$ making asymptotic risk uninteresting.  To avoid this, we will consider asymptotically local alternatives (e.g.
Newey and McFadden 1994).  Specifically, for any fixed~$n$, we consider population heterogeneity that
results in
\begin{equation}
\label{eq:pophetero}
\mathbb{E}_I[\psi(Y_i, X_i; \theta_{0,n})] = n^{-1/2} \delta.    
\end{equation}
where~$\delta$ controls the degree of population heterogeneity. For fixed~$\delta$, the heterogeneity is assumed to disappear asymptotically; however, since there is no restriction on this parameter,~\eqref{eq:pophetero} represents realistic potential heterogeneity between internal and external studies. The parameter $\theta_{0, n}$ depends on $n$ because~\eqref{eq:pophetero} represents a sequence of external study distributions indexed by $n$ that are local to the internal study distribution~\citep{Vaart_1998}.  We next state our main result that demonstrates asymptotically smaller risk compared to the internal-only estimator. 

\begin{theorem}
\label{thm:asympriskfullgeneral}
    Under asymptotically local alternatives given by~\eqref{eq:pophetero}, let
    $$
    J := \Sigma_{\theta}^{-1/2} \Sigma_{\theta, \gamma} A \Sigma_{\gamma, \theta} \Sigma_{\theta}^{-1/2}
    $$
    and~$d := tr(J)/\|J\|$ be the ratio of the trace of $J$ and the largest eigenvalue of $J$.  Then if~$d > 2$, for any~$\tau$ such that $0 < \tau \leq 2 \{ tr(J) - \|J\| \}$, 
    $$
    R(\hat \gamma_{JS}, \gamma^\star) \leq 
    R(\hat \gamma_I, \gamma^\star) - \tau \times \frac{2 \{ tr(J) - 2 \| J \| \} - \tau}{E \left[ (\Delta^\star + \tilde \delta)^\top B (\Delta^\star + \tilde \delta) \right]}.
    $$
    where~$\tilde \delta := -(M_{IE} \delta/\sqrt{c}, 0)^\top$ with~$M_{IE} = Q_{\theta,I}Q_{\theta, E}^{-1}$, $\Delta^\star \sim N(0, W^\star)$, $B = L^\top A L$, $L = \Sigma_{\gamma, \theta} \Sigma_{\theta}^{-1} P_\theta H Q_I^{-1}$ with $P_\theta = (I, 0)$ extracts the $\theta$-component of~$(\theta, \gamma)$, 
    $$
    H := \left( \begin{array}{cc}
         \nabla h(\theta)^\top & {\bf 0}_{p \times q}  \\
         {\bf 0}_{q \times p} & I_{q\times q}
    \end{array} \right), \quad \text{and} \quad
     W^\star  := \left( \begin{array}{cc}
         \frac{1}{c} M_{IE} W_{\theta, E} M_{IE}^\top + W_{\theta I} & W_{\theta, \gamma, I} \\
          W_{\theta, \gamma, I} & W_{\gamma, I}
    \end{array} \right).
    $$
\end{theorem}


Theorem~\ref{thm:asympriskfullgeneral} states that for $d > 2$, $\tau$ ranging from $0$ to $2 \{ tr(J) - 2 \|J\|\}$ yield James-Stein estimators with asymptotically smaller risk than the internal-only estimator.  The necessary condition~$d > 2$ can be seen as a general condition which is satisfied when $A$ is full rank and the dimension of $\theta$ exceeds $2$.  The optimal choice of $\tau$, denoted~$\tau^\star$, is $tr(J) - 2 \| J \|$.  Proof of Theorem~\ref{thm:asympriskfullgeneral} is provided in Appendix~\ref{appendix:theory}.  Note that this result does not assume the identical distributions of external and internal data.  The risk will depend on the differences in the sandwich covariance matrices~$Q$ and $W$. 

The estimate $\hat \tau$ in~\eqref{eq:jsweight} uses the empirical estimate~$\hat J$, which is computed using the sandwich variance estimates~$\hat Q$ and~$\hat W$ respectively.  Then the weights can be re-expressed as
$$
\hat w = \left( 1 - \frac{tr(\hat J) - 2 \| \hat J \|}{\left \{ \sqrt{n} \hat \Sigma_{\theta}^{-1/2} (\hat \theta_E - \hat \theta_I) \right \}^\top \hat J  \left \{ \sqrt{n} \hat \Sigma_{\theta}^{-1/2} (\hat \theta_E - \hat \theta_I) \right \}} \right)_{+}
$$
which shows that the weight is a function of~$\hat J$ and the standardized asymptotic difference between $\hat \theta$ for the external and internal datasets.



\subsection{Bootstrap inference}
\label{sec:boostrap}

In practice, the analyst is not only interested in estimators with lower risk, but also performing statistical inference with these estimators.  It is well known that James-Stein estimators are non-regular which means standard Z-estimation asymptotics do not apply. In this case, the nonparametric bootstrap is consistent on all but a small subset of the underlying parameter space. Modified versions of the bootstrap, such as the m-out-of-n bootstrap and the oracle bootstrap, try to solve the inconsistency of the nonparametric bootstrap under a fixed parameter setting. Recent work studied the local asymptotic behavior of the estimators and of their bootstrap distributions~\cite{10.3150/22-BEJ1538}. Similar arguments can be applied in our local asymptotic setting~\eqref{eq:pophetero} and imply dependence of the limiting distribution on~$\delta$ at non-regular points.  Here, we leverage prior work on generalized bootstrap for estimating equations~\cite{10.1214/009053604000000904} to propose a generalized bootstrap procedure for our James-Stein estimator. In Section~\ref{section:simulations}, we provide empirical results examining the finite sample local performance of the bootstrap estimators. 

First, we consider the external estimates~$(\hat \theta_E, \hat \Sigma_\theta^E)$ fixed.  Then the $k$th bootstrap internal estimator, denoted~$(\hat \theta_I^{(k)}, \hat \gamma_I^{(k)})$, is obtained by solving a weighted set of estimating equations:
$$
\sum_{i=1}^{n_I} \omega_i^{(k)} \left( \begin{array}{c} \psi(Y_i, X_i; \theta) \\ \phi(Y_i, X_i, Z_i; \gamma) \end{array} \right) = 0
$$ 
where~$\omega^{(k)}_i$ satisfies the conditions in~\cite{10.1214/009053604000000904}. In our simulations, we consider~$\omega_i^{(k)} \sim \text{Exp} (1)$. Then the $k$th bootstrap conditional estimator is obtained by applying the conditional estimator defined in~\eqref{eq:newestimator} to the bootstrap internal data. Then the $k$th bootstrap James-Stein estimator is obtained by applying the James-Stein procedure to the bootstrap estimates, i.e., $\hat \theta_{JS}^{(k)} = \hat w^{(k)} \hat \gamma_{\text{cond}}^{(k)} + (1 - \hat w^{(k)}) \hat \gamma_I^{(k)}$
with weight
$$
\hat w^{(k)} = \left( 1 - \frac{tr(\hat J) - 2 \| \hat J \|}{\left \{ \sqrt{n} \hat \Sigma_{\theta}^{-1/2} (\hat \theta_E - \hat \theta_I^{(k)}) \right \}^\top \hat J  \left \{ \sqrt{n} \hat \Sigma_{\theta}^{-1/2} (\hat \theta_E - \hat \theta_I^{(k)}) \right \}} \right)_{+}
$$
where~$\hat J$ is the empirical estimate of~$J$ and~$\hat \Sigma_\theta$ is the empirical estimate of~$\Sigma_\theta^{(I)}$ from the internal data.  Confidence intervals can then be constructed using the empirical bootstrap distribution of~$\{ \hat \theta_{JS}^{(k)}\}$.

\section{Simulations}
\label{section:simulations}

We conduct comprehensive simulation studies to evaluate the performance of our proposed conditional and James-Stein estimators across four distinct scenarios: linear models, logistic regression for binary outcomes, conditional average treatment effects (CATE), and surrogate endpoints. Each simulation demonstrates the utility of our approach in different data integration contexts.

\subsection{Linear regression}
\label{sec:linear}

We begin with linear regression to establish baseline performance of our proposed estimators. The simulation study was based on prior work by~\cite{10.1093/biomtc/ujae072}. The internal study uses covariates $X = (X_1, X_2, X_3, X_4, X_5)$ and auxiliary variables $Z = (Z_1, Z_2)$, where $X_1 \sim \text{Exp}(1)$, $(X_2, \tilde{X}_3, X_4, X_5)$ follows a multivariate normal distribution with mean 0, unit variance, and correlation 0.3, $X_3 = I(\tilde{X}_3 > 0.7 X_2)$, $Z_1 \sim N(0, 1)$, and $Z_2 \sim N(\alpha \log(X_1), 1)$. The response is generated as $Y = \beta_c + \beta_X (\sum_{j=1}^5 X_j + X_1 \cdot X_3) + \beta_Z (Z_1 + Z_2) + \beta_{XZ} X_2 \cdot Z_2 + \varepsilon$ with $\varepsilon \sim N(0, 4)$. We set $\alpha = 0.2$, $\beta_c = 0.5$, $\beta_X = 0.5$, $\beta_Z = 0.2$, and $\beta_{XZ} = 0.2$ for the internal study with sample size $n_I = 200$. The external study uses covariates $(X_1, X_2, X_4, X_5)$ with sample size $n_E = 20,000$. We introduce systematic bias in the external study relative to the internal study by applying offsets to the $\beta_X$ parameters in the internal study, ranging from 0 to 0.3 in increments of 0.025. This creates scenarios where the external and internal studies have varying degrees of population heterogeneity.

We see that the predictive mean square error (PMSE) is smaller for the conditional estimator than the internal-only estimator when the bias is small, while the James-Stein estimator has the guaranteed reduction regardless of the level of bias.  We include the GIM~\cite{10.1093/biomet/asaa014} as a comparison as it is the most recent method that can be used to integrate external summary data with a newly conducted internal study using constrained maximum likelihood that accounts for the uncertainty in the estimates from the external study.  We see equivalent performance to the conditional estimator as the model-based assumptions are correct.

\begin{figure}[!th]
    \centering
    \begin{subfigure}[t]{0.48\textwidth}
        \centering
        \includegraphics[width=\textwidth]{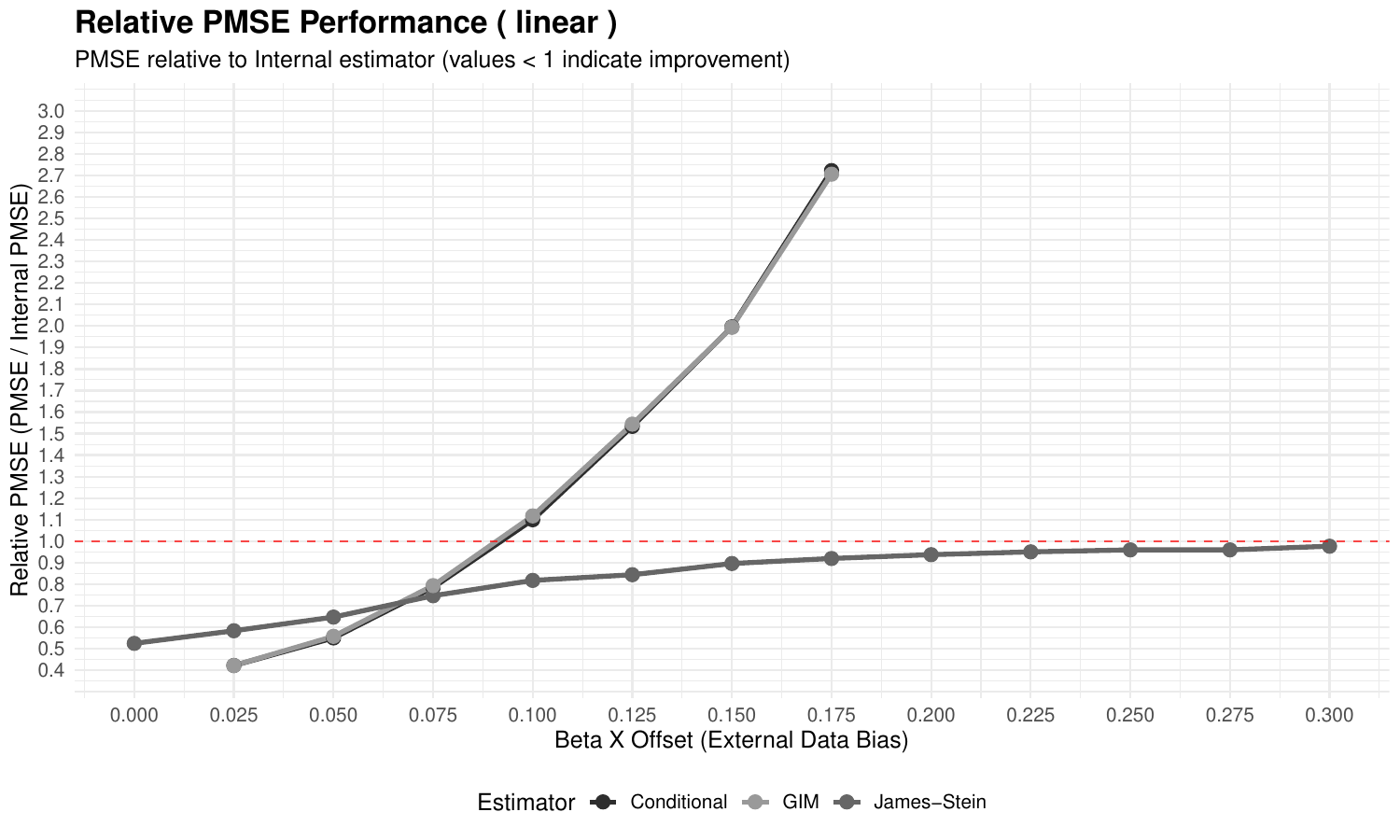}
        \caption{Relative PMSE comparison}
        \label{fig:linear_simulation_pmse}
    \end{subfigure}%
    \hfill
    \begin{subfigure}[t]{0.48\textwidth}
        \centering
        \includegraphics[width=\textwidth]{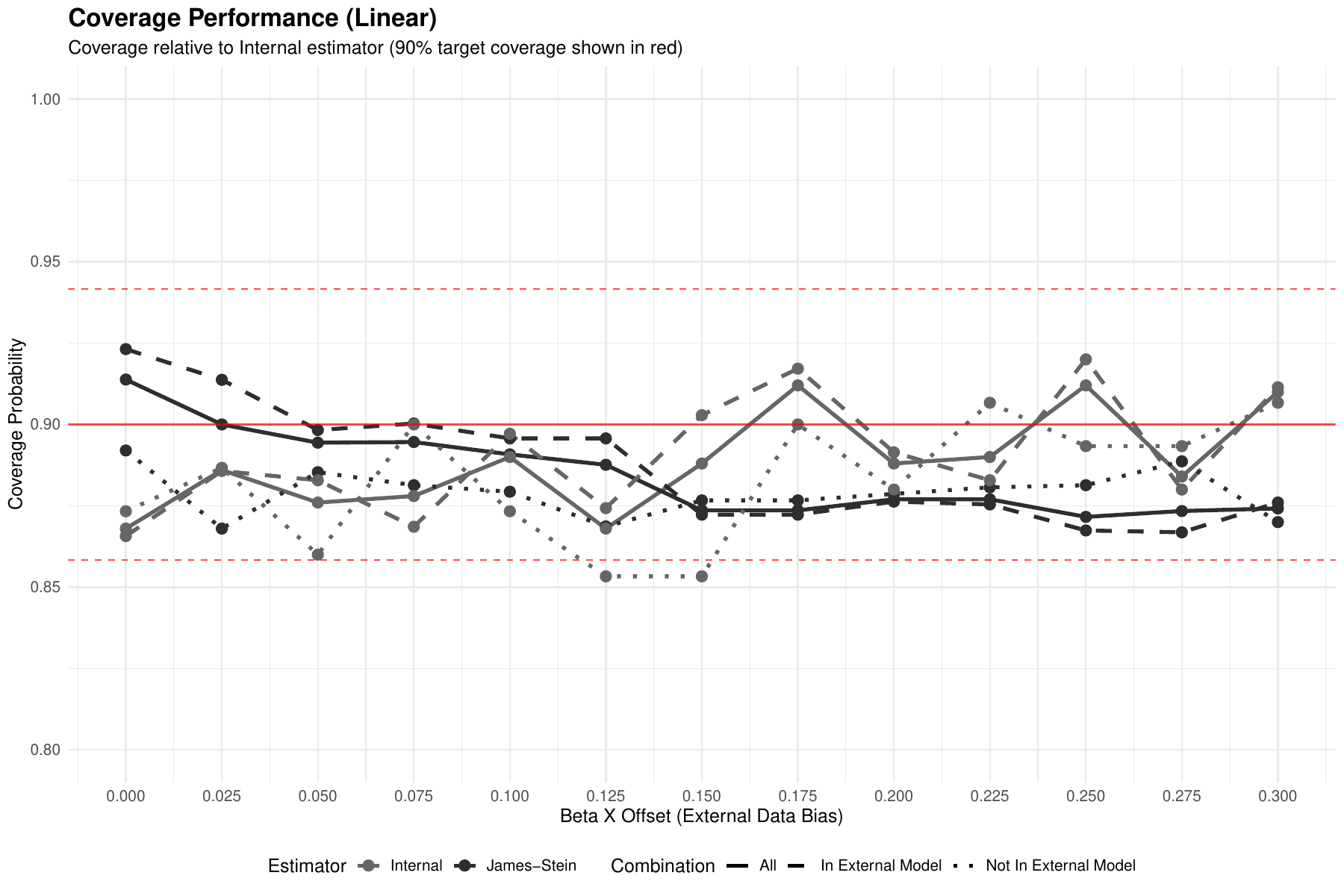}
        \caption{Coverage rate comparison}
        \label{fig:linear_simulation_coverage}
    \end{subfigure}
    \caption{Relative PMSE (left) and mean coverage rate (right) for linear regression across varying external bias.}
    \label{fig:linear_simulation_combined}
\end{figure}

The coverage plot for the linear regression scenario (right panel) shows the empirical coverage rates of the 90\% confidence intervals across varying levels of external bias. We present the mean coverage rate across 200 simulations, with one curve for the average across all parameters, one curve for the average across paramaters in the external model, and one curve for the average across paramaters not in the external model. Coverage remains close to the nominal 90\% level for small offsets, with a slight decline as the bias increases. This indicates that the proposed methods maintain appropriate uncertainty quantification when the external and internal populations are similar, but coverage can be impacted by the level of heterogeneity.

\subsection{Logistic regression}
\label{sec:logistic}

We next consider logistic regression for binary outcomes, where we employ the same covariate structure as the linear simulation but generate the response using a logistic model: $\text{logit}(P(Y = 1)) = \beta_c + \beta_X (\sum_{j=1}^5 X_j + X_1 \cdot X_3) + \beta_Z (Z_1 + Z_2) + \beta_{XZ} X_2 \cdot Z_2$. This setting is particularly challenging due to the non-collapsibility of logistic regression, which can limit the effectiveness of external information for certain parameters.  Again the predictive mean square error (PMSE) is smaller for the conditional estimator than the internal-only estimator when the bias is small, while the James-Stein estimator has the guaranteed reduction regardless of the level of bias.  We again include the GIM as a comparison as that has similar performance to the conditional estimator since the model-based assumptions are correct.  Additionally we include Firth corrected estimating equations as a comparison.  We see that the Firth corrected approaches have improved predictive mean square error.  

\begin{figure}[!th]
    \centering
    \begin{subfigure}[t]{0.48\textwidth}
        \centering
        \includegraphics[width=\textwidth]{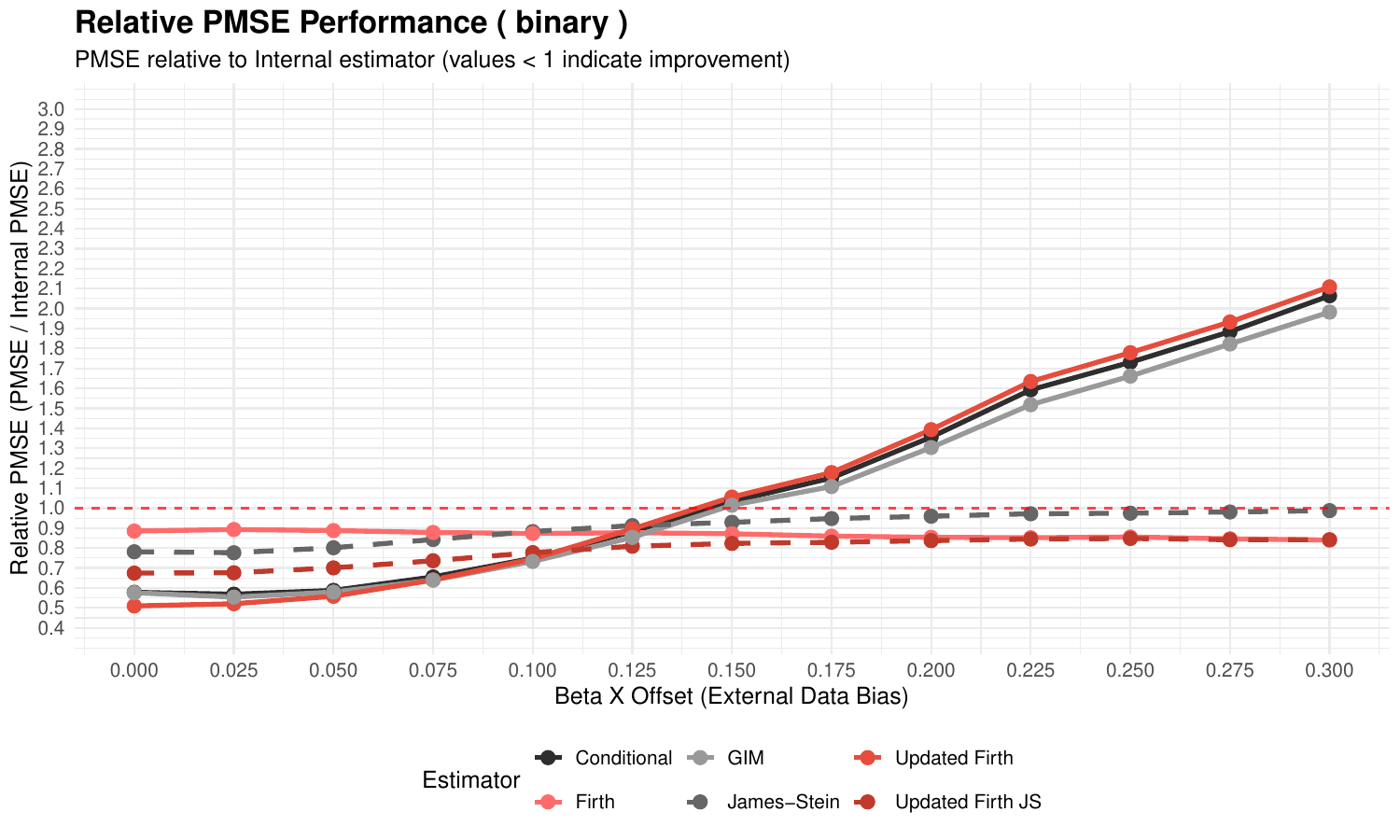}
        \caption{Relative PMSE comparison}
        \label{fig:binary_simulation_pmse}
    \end{subfigure}%
    \hfill
    \begin{subfigure}[t]{0.48\textwidth}
        \centering
        \includegraphics[width=\textwidth]{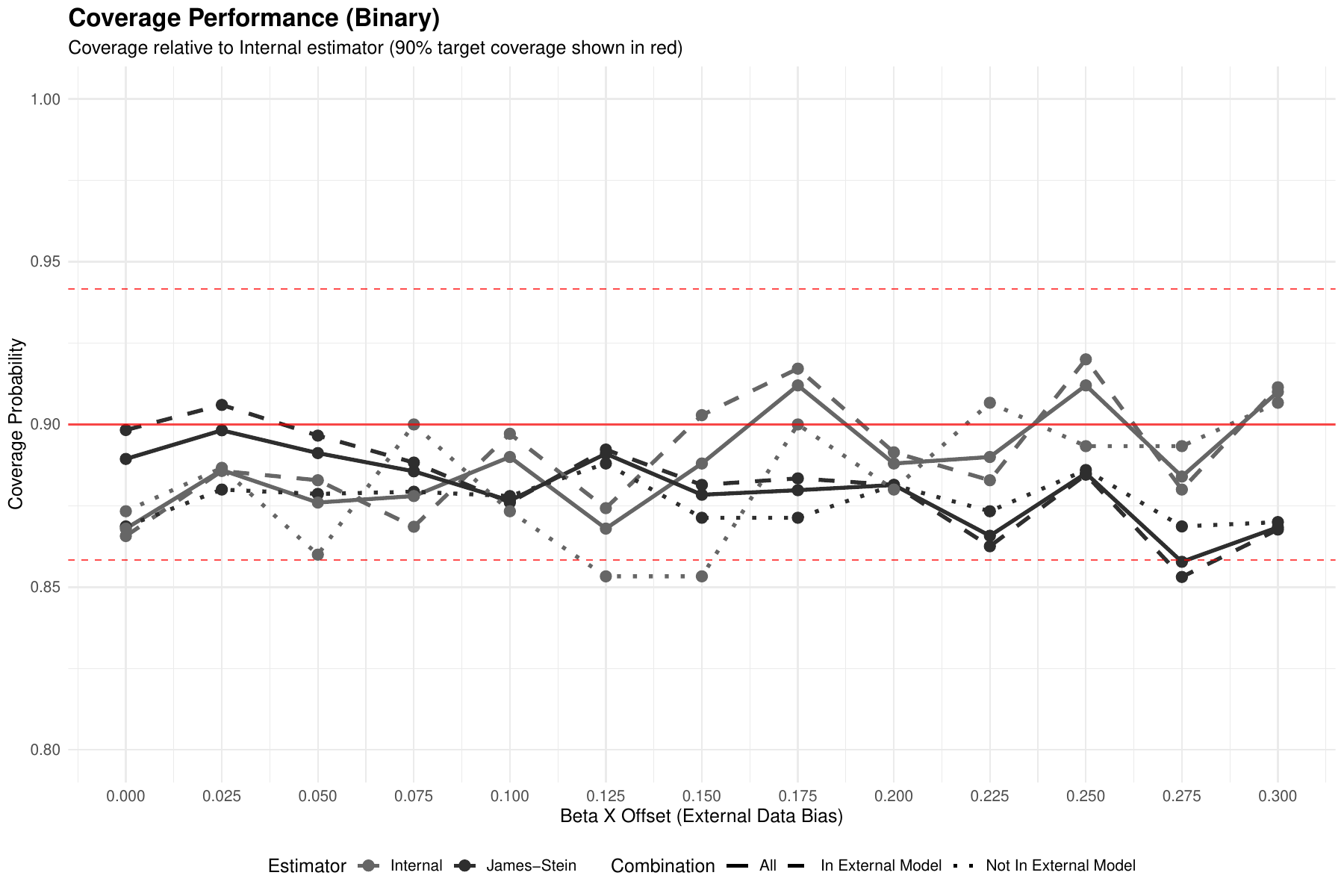}
        \caption{Coverage rate comparison}
        \label{fig:binary_simulation_coverage}
    \end{subfigure}
    \caption{Relative PMSE (left) and mean coverage rate (right) for logistic regression across varying external bias.}
    \label{fig:binary_simulation_combined}
\end{figure}

The coverage plot for the binary outcome scenario (right panel) demonstrates that the 90\% confidence intervals generally achieve nominal coverage across the range of external bias. The coverage is stable and close to 90\%, with only minor fluctuations, suggesting that the methods provide reliable interval estimates even in the presence of moderate population heterogeneity.

\subsection{Conditional average treatment effect (CATE) simulation}
\label{sec:cate}

Our next simulation evaluates the performance of our proposed estimators in causal inference settings with heterogeneous treatment effects. We generate data with binary treatment $A \in \{0,1\}$ and outcome $Y$, where the treatment assignment follows a logistic model: $\text{logit}(P(A = 1)) = \alpha_0 + \alpha_1 X_1 + \alpha_2 Z_1$. The outcome model includes both main effects and treatment interactions: $Y = \beta_0 + \sum_{j=1}^5 \beta_j X_j + \sum_{k=1}^2 \beta_{5+k} Z_k + \tau_0 A + \sum_{j=1}^2 \tau_j X_j A + \tau_3 Z_1 A + \tau_4 X_2 Z_2 A + \varepsilon$. In this example, the conditional average treatment effect is the same across internal and external studies; however, the terms that do not involve interaction with treatment indicators do exhibit heterogeneity.  
The $\psi$ estimating equations are weighted linear regressions based on equation~\eqref{eq:wcls} with control variables $g(X,Z)=(1,X_1,\ldots, X_5, Z_1, Z_2)))$ and treatment model $f(X,Z) = (1,X_1,\ldots, X_5, Z_1, X_2 Z_2)$.  The $\phi$ estimating equations also follow from equation~\eqref{eq:wcls} but now with control variables $g(X)=(1,X_1,\ldots, X_5)$ and treatment model $f(X) = (1,X_1,\ldots, X_5)$.  
In this case, Assumption~\ref{consistency} holds for the conditional treatment effect. 
We choose the weighted quadratic loss to be the predictive mean square error for the CATE.  Consistent with the fact that Assumption~\ref{consistency} holds across all scenarios, Figure~\ref{fig:cate_simulation_pmse} shows that the conditional estimator has the smallest predictive mean square error for the CATE regardless of the heterogeneity.  

\begin{figure}[!th]
    \centering
    \begin{subfigure}[t]{0.48\textwidth}
        \centering
        \includegraphics[width=\textwidth]{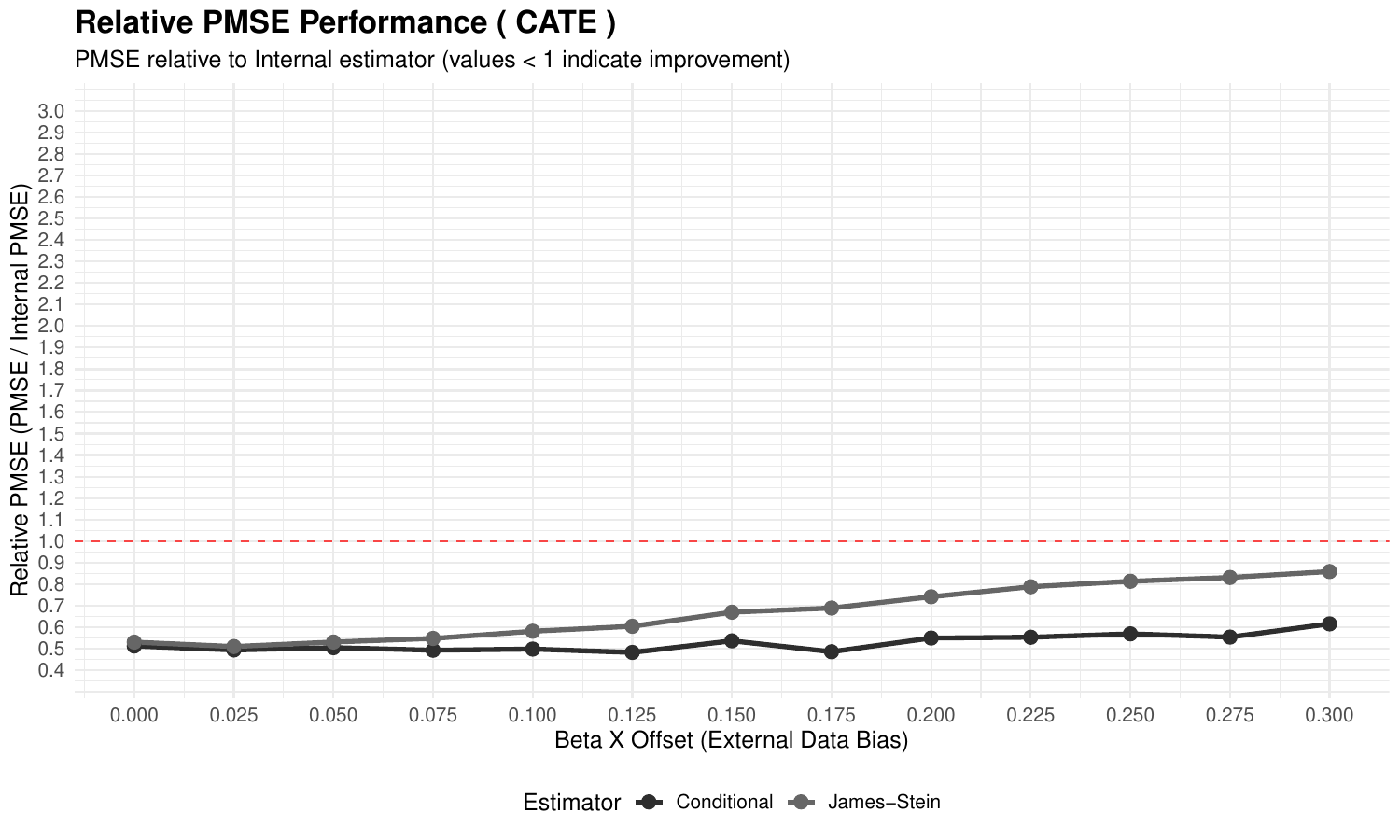}
        \caption{Relative Predictive CATE MSE comparison}
        \label{fig:cate_simulation_pmse}
    \end{subfigure}%
    \hfill
    \begin{subfigure}[t]{0.48\textwidth}
        \centering
        \includegraphics[width=\textwidth]{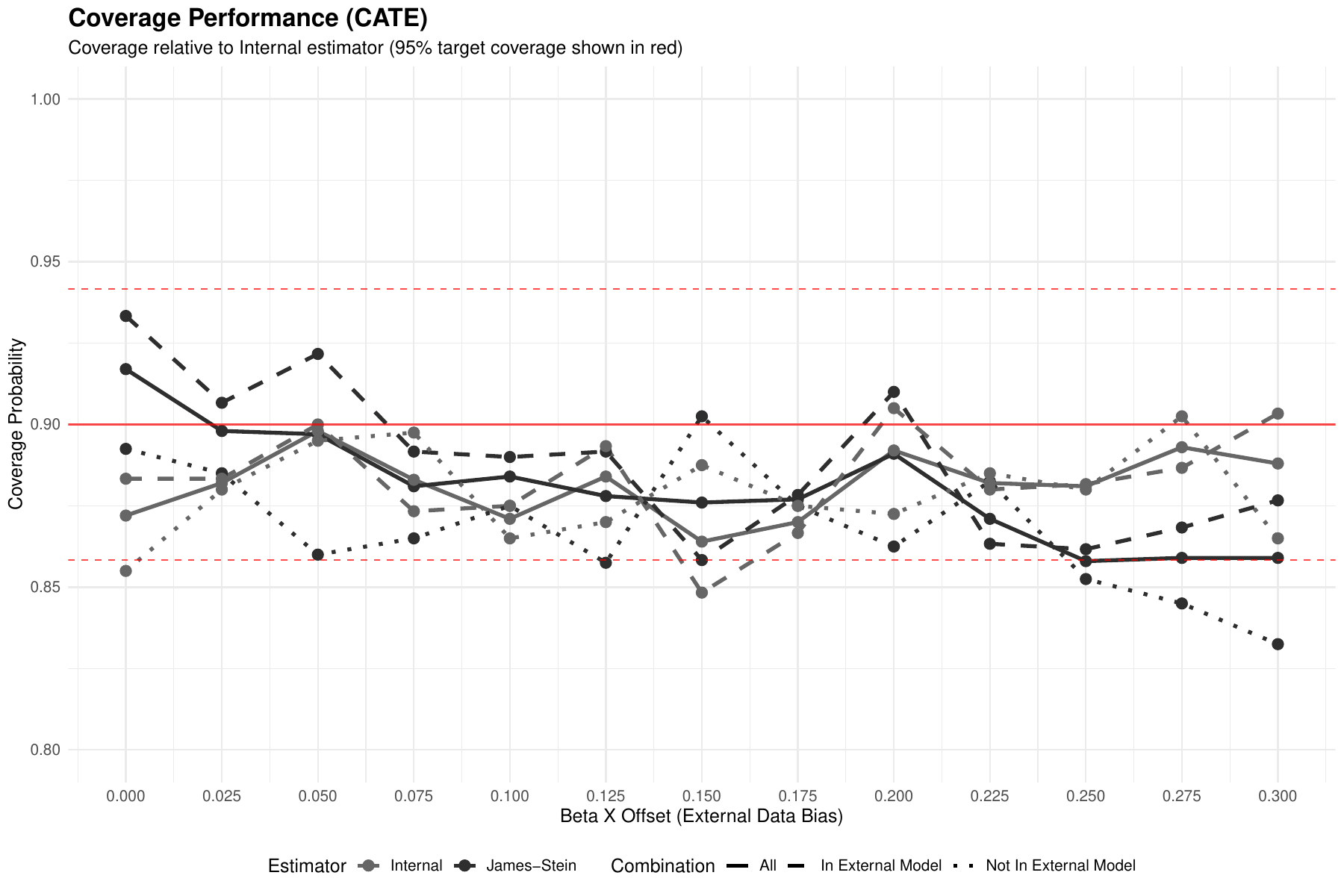}
        \caption{Coverage rate comparison}
        \label{fig:cate_simulation_coverage}
    \end{subfigure}
    \caption{Relative Predictive CATE MSE (left) and mean coverage rate (right) for CATE estimation across varying external bias.}
    \label{fig:cate_simulation_combined}
\end{figure}

The coverage plot for the CATE scenario (right panel) indicates that the empirical coverage of the 90\% confidence intervals is slightly below nominal for larger offsets, but remains reasonably close to the target level across most scenarios. This suggests that the methods provide adequate uncertainty quantification for heterogeneous treatment effect estimation, though some undercoverage may occur as external bias increases.

\subsection{Surrogate endpoint simulation}

The surrogate endpoint simulation explores scenarios where the external study measures a different outcome than the internal study. We generate bivariate outcomes $(Y_1, Y_2)$ from a multivariate normal distribution with the same mean as in mean model as in Section~\ref{sec:linear} and covariance matrix $\Sigma$ with variances $\sigma_1^2, \sigma_2^2$ and correlation $\rho$. The internal study measures both outcomes $(Y_1, Y_2)$ while the external study only measures $Y_2$. Our primary interest is in estimating the regression parameters when $Y_1$ is the outcome of interest, but we leverage external information about the surrogate to improve efficiency.

We conduct simulations across four correlation levels: $\rho \in \{0.7, 0.8, 0.9, 1.0\}$, representing varying degrees of association between the primary and surrogate outcomes. Higher correlations indicate stronger surrogate relationships, which should lead to greater efficiency gains from external information.  Figure~\ref{fig:surrogate_simulation} shows the results, which demonstrate that the conditional and James-Stein estimators achieve PMSE gains. Results show that efficiency gains increase with the correlation between outcomes, as expected. When $\rho = 1.0$, the surrogate is perfectly correlated with the primary outcome, leading to maximum efficiency gains and mirroring the linear regression simulations from Section~\ref{sec:linear}. Even more moderate correlations ($\rho = 0.7$) provide meaningful improvements, demonstrating the practical utility of our approach in settings with imperfect surrogates. Coverage for the surrogate endpoint scenario is not shown, as it is equivalent to the linear case and does not provide additional insight.

\begin{figure}[!th]
     \centering
     \includegraphics[width=0.8\textwidth]{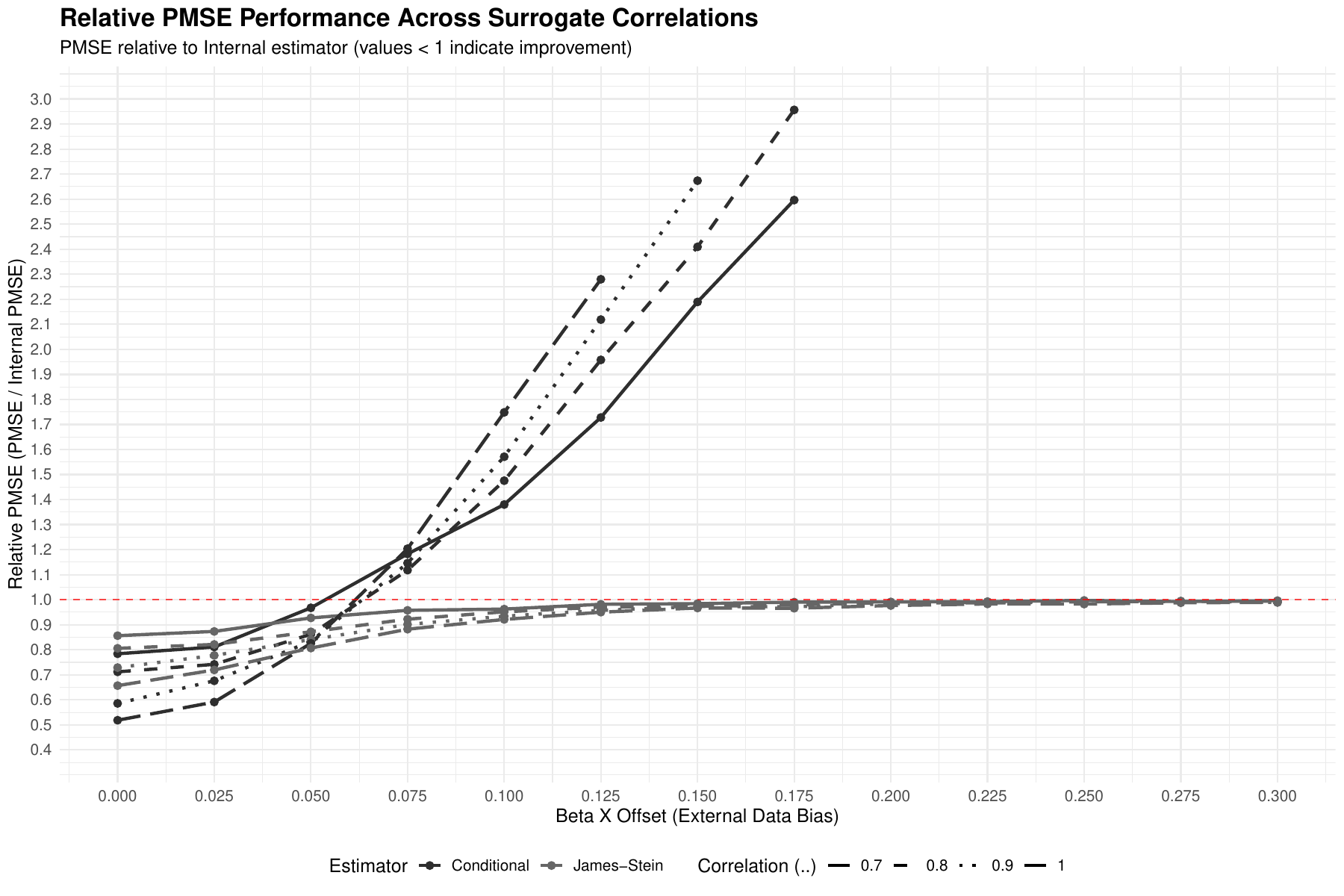}
     \caption{Relative PMSE comparison between internal-only, conditional, and James-Stein estimators for surrogate endpoint.}
     \label{fig:surrogate_simulation}
\end{figure}

\section{Intern Health Study: A Case Study}

The Intern Health Study (IHS) is a 6-month micro-randomized trial on medical interns (NeCamp et al., 2020), which aimed to investigate when to provide mHealth interventions to individuals in stressful work environments to improve their behavior and mental health. The study annually recruits medical interns who are about to start their internship program.  In this paper, we focus on data collected in the 2023 and 2024 cohorts, consisting of 1253 and 859 individuals, respectively.  The analyses conducted in this paper focus on daily randomization to receive a potential afternoon message with probability 1/2. We are interested in assessing the effectiveness of targeted notifications in improving proximal physical activity. To this end, we focus on step count in the 24 hours after randomization as the longitudinal outcome of interest.

We start by assuming the 2023 IHS data is the external study and use the causal analysis discussed in Example~\ref{ex:cate} to estimate the potentially time-varying causal effect of daily messaging on daily step count.  We assume the external model $\phi$ includes the following control variables: Neuroticism (baseline), Steps (3 Hours Before), Weekend Indicator, and a linear spline model over days in study with knots at 30 and 60. All variables except Neuroticism are used in the time-varying treatment effect model.  We then assume the 2024 IHS data is the internal study and again use the causal analysis discussed in Example~\ref{ex:cate}.  The internal model $\psi$ uses the same set of control variables and treatment effect model but includes a new variable ``Minutes Asleep (Last 24 Hours)'' which was not available in the 2023 data. Figure~\ref{fig:both2023and2024} shows the causal effect estimates for internal only, conditional, and James-Stein approaches. We see that the weights are close to 1 and therefore, the James-Stein estimator is close to the internal only estimator in this setting.  

The previous results are driven by the differences among the two cohorts as well as the large sample size of the 2024 study. To further evaluate our method in settings with larger external studies and smaller internal studies, we consider the 2024 IHS data alone and split it into synthetic external and internal studies by randomly splitting individual-level data (90\% external and 10\% internal respectively).  This mimics the data integration setting for micro-randomized trials more broadly~\cite{huch2024dataintegrationmethodsmicrorandomized}. Figure~\ref{fig:only2024} shows the causal effect estimates for internal only, conditional, and James-Stein approaches in this setting. We see that the weights are no longer close to 1 and the James-Stein estimator is a weighted average of the internal only and conditional estimators.  
\begin{figure}[htbp]
\centering
\begin{subfigure}[b]{0.48\textwidth}
\centering
\includegraphics[width=\textwidth]{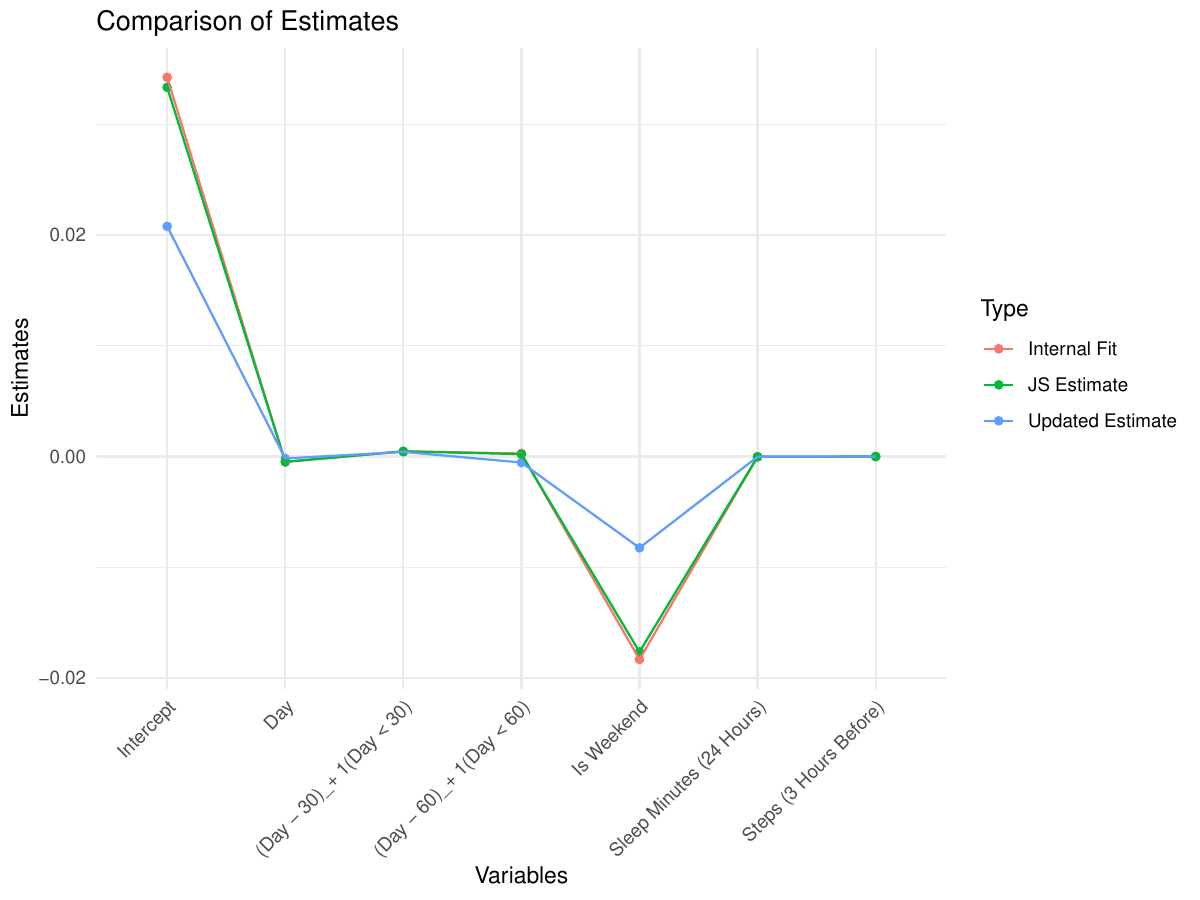}
\caption{Combined 2023 and 2024 estimates}
\label{fig:both2023and2024}
\end{subfigure}
\hfill
\begin{subfigure}[b]{0.48\textwidth}
\centering
\includegraphics[width=\textwidth]{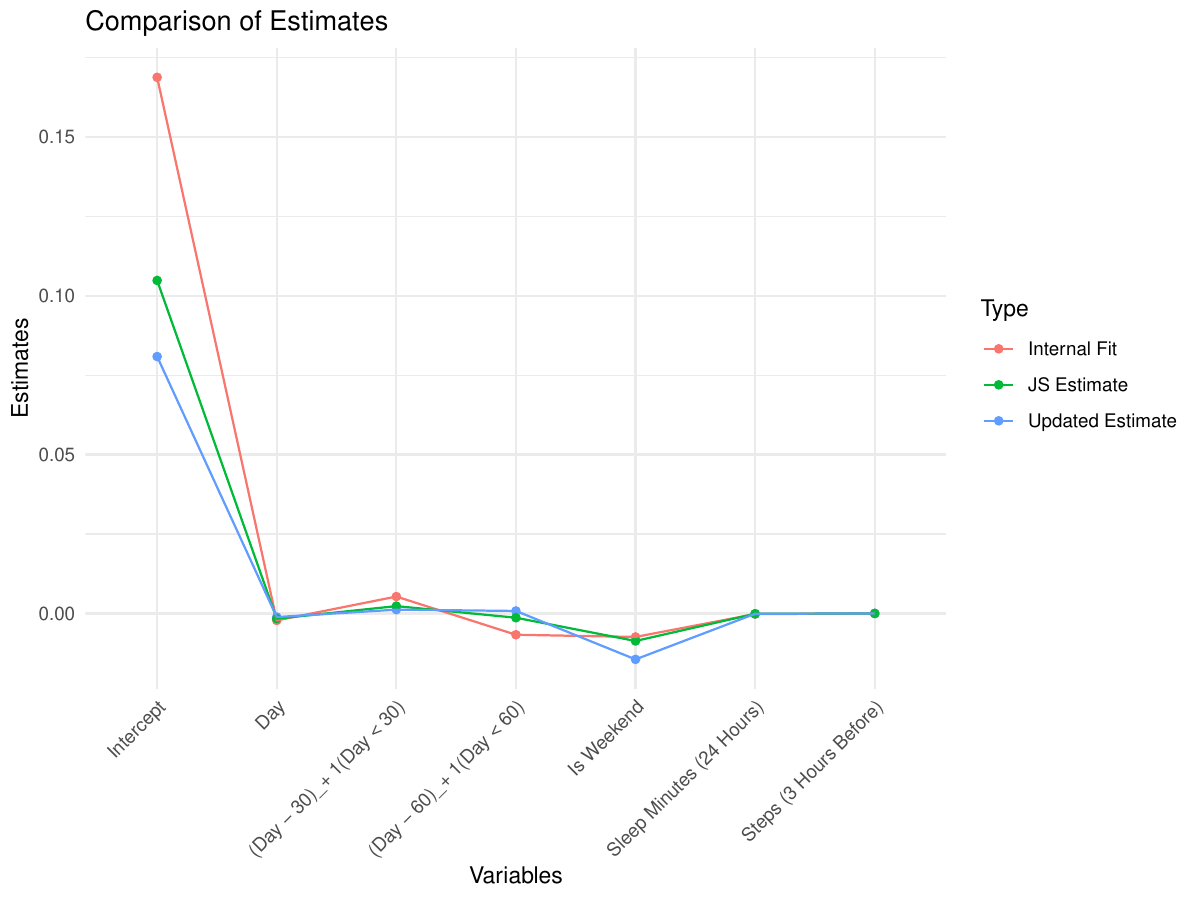}
\caption{2024 estimates only}
\label{fig:only2024}
\end{subfigure}
\caption{Estimation results for the Intern Health Study case study showing treatment effect estimates across different time periods.}
\label{fig:ihs_estimates}
\end{figure}

Based on Figure~\ref{fig:only2024}, we then consider inference using the bootstrap approach discussed in Section~\ref{sec:boostrap}.  Figure~\ref{fig:cibs_simulation} focuses on the confidence intervals for the intercept and weekend indicators of the treatment effect model. We see that the internal only estimator has the widest confidence intervals, while the conditional estimator has the narrowest intervals.  The James-Stein estimator is a weighted average of the two, and this is also reflected in the width of the confidence intervals. 

\begin{figure}[!ht]
    \centering
    \includegraphics[width=0.7\textwidth]{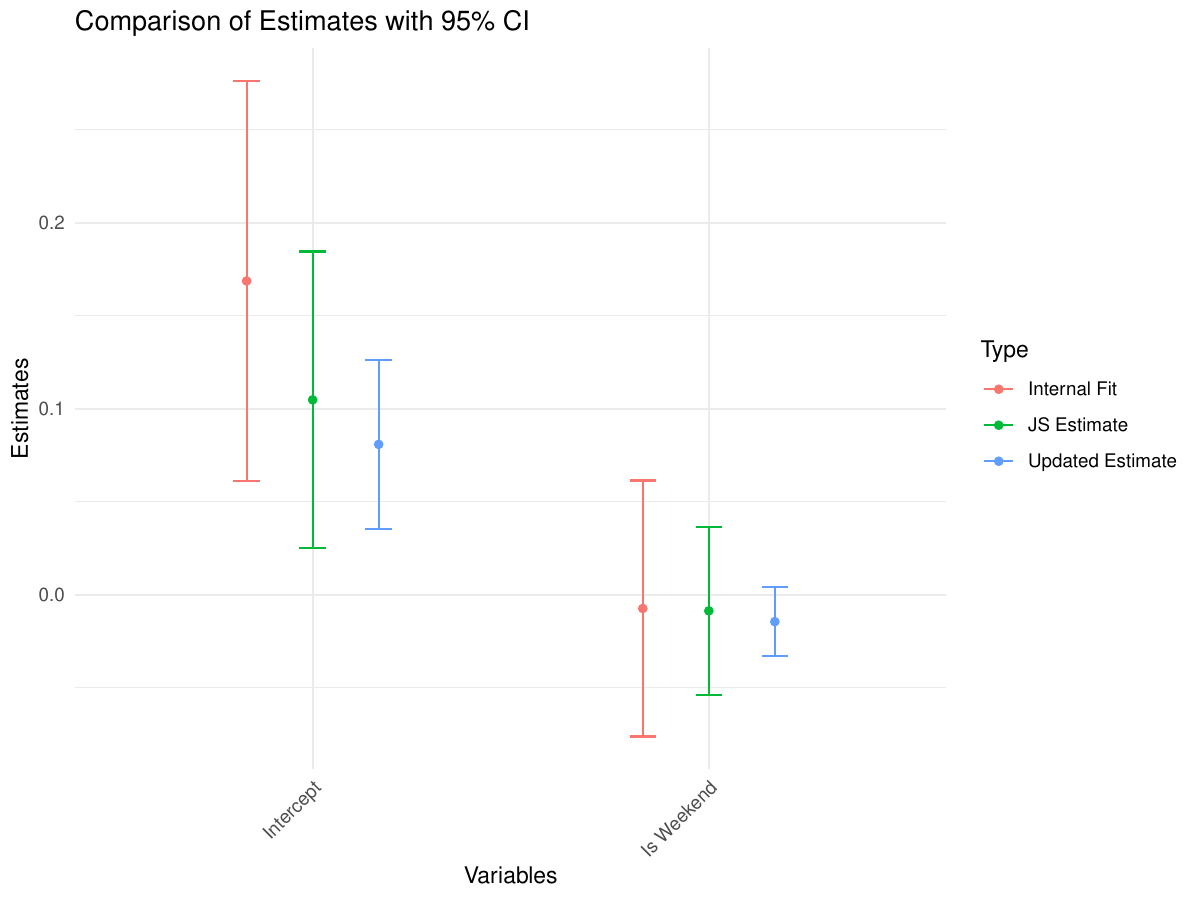}
    \caption{Confidence Intervals for the intercept and weekend indicators of the treatment effect model based on the bootstrap approach discussed in Section~\ref{sec:boostrap}.}
    \label{fig:cibs_simulation}
\end{figure}

\appendix

\section{Appendix 1}
\label{appendix:theory}

\begin{proof}[Proof of Theorem~\ref{thm:glms}]
    By Proposition 2 in Dai et al. (2016), the maximum likelihood estimators $(\hat \theta_0,\hat \theta_X)$ and $\hat \beta_Z$ are asymptotically independent for the internal dataset.  
    This implies that $\Sigma^{I}_{\theta, \gamma_Z}$ is zero.  Therefore,~$\Sigma_{\gamma_Z,\theta}^h \left( \Sigma_\theta^h \right)^{-1}$ is zero, resulting in $\hat \gamma_{\cond, Z} = \hat \gamma_{I,Z}$ as desired.
\end{proof}

\begin{proof}[Proof of Theorem~\ref{thm:asympriskfullgeneral}]
We start by showing that the asymptotic risk for the internal-only estimator~$\hat \gamma_I$ is
$$
R(\hat \gamma_I, \gamma^\star) = E \left( \Delta_{\gamma}^\top Q^{-1}_{\gamma} A Q^{-1}_{\gamma} \Delta_{\gamma} \right) 
$$
where~$\Delta_{\gamma} \sim N(0, W_{\gamma})$. To do so, we leverage the following result from~\cite{HANSEN2016115}.
\begin{lemma}[Lemma 1~\cite{HANSEN2016115}]
\label{lemma:hansen}
    For any estimator $\hat \theta$ satisfying $\sqrt{n} (\hat \theta - \theta) \to^d \Delta$ as $n \to \infty$ for some random variable~$\Delta$, and for the weighted quadratic loss $l(\hat \theta, \theta^\star) = (\hat \theta - \theta^\star)^\top A (\hat \theta - \theta^\star)$, the asymptotic risk for $\hat \theta$ is $R(\hat \theta, \theta^\star) = E \left( \Delta^\top A \Delta \right)$.
\end{lemma}
For the estimating equations, we have 
$$
\sqrt{n_I} ( \hat \gamma_I - \gamma^\star) = Q_{\gamma,I}^{-1} \frac{1}{\sqrt{n_I}} \sum_{i=1}^{n_I} \phi (Y_i, X_i, Z_i; \gamma_0) + o_p (1) \to^d Q_{\gamma,I}^{-1} \Delta_{\gamma,I}
$$
where~$\Delta_{\gamma,I} \sim N(0, W_{\gamma,I})$ where
$$
    W_{\gamma,I} := \mathbb{E}_I \left[ \phi (Y, X, Z; \gamma_0) \phi (Y, X, Z; \gamma_0)^\top \right], \quad
    Q_{\gamma,I} := \mathbb{E}_I \left[ \frac{\partial}{\partial \gamma} \phi (Y, X, Z; \gamma_0) \right],
$$
Then the risk follows by Lemma~\ref{lemma:hansen}.

For the conditional estimator, we start with a similar Taylor series argument.  
    \begin{equation*}
    \begin{split}
    0 =&\frac{1}{\sqrt{n_E}} \sum_{i=1}^{n_E} \psi(Y_i,X_i; \theta_0) + 
    \frac{1}{{n_E}} \sum_{i=1}^{n_E} \frac{\partial}{\partial \theta} \psi(Y_i,X_i; \bar \theta) \sqrt{n_E} \left( \hat \theta_E - \theta_0 \right)    \\
    0 =&\frac{1}{\sqrt{n_I}} \sum_{i=1}^{n_I} \psi(Y_i,X_i; \theta_0) + 
    \frac{1}{{n_I}} \sum_{i=1}^{n_I} \frac{\partial}{\partial \theta} \psi(Y_i,X_i; \bar \theta) \sqrt{n_I} \left( \hat \theta_I - \theta_0 \right)    \\
    \Rightarrow \sqrt{n_I} \left( h(\hat \theta_I) - h(\theta_0)\right) &= \nabla h(\theta)^\top \left( Q_{\theta, I} \right)^{-1} \frac{1}{\sqrt{n_I}} \sum_{i=1}^{n_I} \psi(Y_i,X_i; \theta_0) + o_p (1) \\
    \Rightarrow \sqrt{n_E} \left( h(\hat \theta_E) - h(\theta_0) \right) &= \nabla h(\theta)^\top \left( Q_{\theta, E} \right)^{-1} \left( \frac{1}{\sqrt{n_E}} \sum_{i=1}^{n_E} ( \psi(Y_i,X_i; \theta_0) - E[ \psi (Y_i, X_i; \theta_0) ]) + \delta \right) + o_p (1) \\
    \Rightarrow \sqrt{n} \left( h(\hat \theta_E) - h(\hat \theta_I) \right) &\to^d \nabla h(\theta)^\top \left \{ \left( \sqrt{c} Q_{\theta, E} \right)^{-1} \left( \Delta_{\theta}^E + \delta \right) - \left( Q_{\theta,I} \right)^{-1} \Delta_{\theta}^I \right \}.
    \end{split}
    \end{equation*}
    where~$\Delta_{\theta}^E \sim N(0, W_{\theta}^E)$ and similarly~$\Delta_{\theta}^I \sim N(0, W_{\theta}^I)$.

Recall we define~$M_{IE} = Q_{\theta,I} Q_{\theta,E}^{-1}$.  Then we can simplify and we have 
    $$ \sqrt{n} ( h(\hat \theta_I) - h(\hat \theta_E) ) \to^d 
    \nabla h(\theta)^\top Q_{\theta, I}^{-1} \left \{ \Delta_\theta^I - \frac{1}{\sqrt{c}} M_{IE} \Delta_{\theta}^E - \frac{\delta}{\sqrt{c}} M_{IE} \right \}
    $$
    Let~
    $$
    \Delta^\star = \left(  \begin{array}{c}
         \Delta_{\theta}^I - \frac{1}{\sqrt{c}} M_{IE} \Delta_{\theta}^E\\
         \Delta_{\gamma}^I 
    \end{array} \right) 
    $$
    The covariance of $\Delta^\star$ is given by
    $$
    W^\star  = \left( \begin{array}{cc}
         \frac{1}{c} M_{IE} W_{\theta, E} M_{IE}^\top + W_{\theta I} & W_{\theta, \gamma, I} \\
          W_{\theta, \gamma, I} & W_{\gamma, I}
    \end{array} \right)
    $$    
    Then we can derive the joint distribution
    $$
    \sqrt{n} \left( 
    \begin{array}{c}
        h(\hat \theta_I) - h(\hat \theta_E) \\
        \hat \gamma_I - \gamma_0
    \end{array}
    \right) \to^d
    \left( \begin{array}{cc}
         \nabla h(\theta)^\top & {\bf 0}_{p \times q}  \\
         {\bf 0}_{q \times p} & I_{q\times q}
    \end{array} \right)
    Q_{I}^{-1} \left \{ \Delta^\star - \left(
    \begin{array}{c}
         \delta/\sqrt{c} M_{IE} \\
         {\bf 0}
    \end{array}
    \right)
    \right\}.
    $$
    We can extract components using~$P_{\theta} = (I_{p \times p},0_{p \times q}) \in \mathbb{R}^{p \times (p+q)}$ and $P_{\gamma} = (0_{q \times p}, I_{q \times q}) \in \mathbb{R}^{q \times (p+q)}$.  Then
    \begin{align*}
    \sqrt{n} (\hat \gamma_{\text{cond}} - \gamma_0) &=
    \sqrt{n} (\hat \gamma_I - \Sigma_{\gamma, \theta} ( \Sigma_{\theta})^{-1} ( \hat \theta_I - \hat \theta_E ) - \gamma_0) \\
    &= \sqrt{n} (\hat \gamma_I - \gamma_0) + \Sigma_{\gamma, \theta} ( \Sigma_{\theta})^{-1} \sqrt{n} ( \hat \theta_E - \hat \theta_I ) \\
    &\to^d P_{\gamma} Q_I^{-1} \Delta^\star - \Sigma_{\gamma, \theta} (\Sigma_{\theta})^{-1} P_\theta \left( \begin{array}{cc}
         \nabla h(\theta)^\top & {\bf 0}_{p \times q}  \\
         {\bf 0}_{q \times p} & I_{q\times q}
    \end{array} \right) Q_I^{-1} \left[ \Delta^\star -
    \left(
    \begin{array}{c}
         M_{IE} \delta/\sqrt{c} \\
         {\bf 0} 
    \end{array}
    \right)
    \right]
    \end{align*}
    Since~$P_\gamma Q_I^{-1} \Delta^\star = Q_{\gamma,I}^{-1} \Delta^I_\gamma$, 
    this implies that 
    $$
    \sqrt{n} (\hat \gamma_{\text{cond}} - \hat \gamma_I) \to^d
    -\Sigma_{\gamma, \theta} (\Sigma_{\theta})^{-1} P_{\theta} 
    \left( \begin{array}{cc}
         \nabla h(\theta)^\top & {\bf 0}_{p \times q}  \\
         {\bf 0}_{q \times p} & I_{q\times q}
    \end{array} \right)
    Q_I^{-1} \left[ \Delta^\star -
    \left(
    \begin{array}{c}
         M_{IE} \delta/\sqrt{c} \\
         {\bf 0} 
    \end{array}
    \right)
    \right].
    $$
    To limit notation, we will define
    $$
    H := \left( \begin{array}{cc}
         \nabla h(\theta)^\top & {\bf 0}_{p \times q}  \\
         {\bf 0}_{q \times p} & I_{q\times q}
    \end{array} \right)
    $$
    This leads to 
    $$
    n (\hat \gamma_{\text{cond}} - \hat \gamma_I)^\top A (\hat \gamma_{\text{cond}} - \hat \gamma_I) \to^d 
    \left[ \Delta^\star - 
    \left(
    \begin{array}{c}
         M_{IE} \delta/\sqrt{c} \\
         {\bf 0} 
    \end{array}
    \right)
    \right]^\top 
    B
    \left[ \Delta^\star - 
    \left(
    \begin{array}{c}
         M_{IE} \delta/\sqrt{c} \\
         {\bf 0} 
    \end{array}
    \right)
    \right] := \xi
    $$
    where~$L = \Sigma_{\gamma, \theta} (\Sigma_{\theta})^{-1} P_{\theta} H Q_I^{-1}$ and $B = L^\top A L$.  
    Therefore, we have $\hat w \to^d w = (1 - \tau/\xi)_{+}$, and thus
    \begin{align*}
    \sqrt{n} ( \hat \gamma_{JS} - \gamma_I ) &\to^d 
    w Q_{\gamma,I}^{-1} \Delta_{\gamma}^I + (1-w) 
    \left( 
    P_{\gamma} Q_I^{-1} \Delta^\star - \Sigma_{\gamma, \theta} (\Sigma_{\theta})^{-1} P_{\theta} H Q_I^{-1} \left[ \Delta^\star - 
    \left(
    \begin{array}{c}
         M_{IE} \delta/\sqrt{c} \\
         {\bf 0} 
    \end{array}
    \right)
    \right]
    \right) := \xi_{JS}
    \end{align*}  
    Define the random variable~$\xi_{JS}^\star$ that does not trim the weights
    \begin{align*}
    \xi_{JS}^\star  &= (1-\tau/\xi) Q_{\gamma,I}^{-1} \Delta^I_{\gamma} + \tau/\xi 
    \left( 
    P_{\gamma} Q_I^{-1} \Delta^\star - \Sigma_{\gamma, \theta} (\Sigma_{\theta})^{-1} P_{\theta} H Q_I^{-1} \left[ \Delta^\star - 
    \left(
    \begin{array}{c}
         M_{IE} \delta/\sqrt{c} \\
         {\bf 0} 
    \end{array}
    \right)
    \right]
    \right) \\
    &= P_{\gamma} Q_I^{-1} \Delta^\star - \tau/\xi 
    L \left[
    \Delta^\star - 
    \left(
    \begin{array}{c}
         M_{IE} \delta/\sqrt{c} \\
         {\bf 0} 
    \end{array}
    \right)
    \right].
    \end{align*}
    By Lemma~\ref{lemma:hansen}, the risks of both estimators are expectations of quadratic forms.  Moreover, by Lemma 2 in Hansen (2015), we have
    $$
    R(\hat \gamma_{JS}, \gamma^\star) = E ( \xi_{JS}^\top A \xi_{JS} )
    < E ( (\xi^\star_{JS})^\top A \xi_{JS}^\star ),
    $$
    so we can study the risk of the untrimmed estimator.
    To simplify notation, define $\tilde \delta := - \left(
    \begin{array}{c}
         M_{IE} \delta/\sqrt{c} \\
         {\bf 0} 
    \end{array}
    \right)$.  Then we have 
    \begin{align*}
        E \left( (\xi_{JS}^\star)^\top A \xi_{JS}^\star \right) =&  
        R( \hat \gamma_I, \gamma_0) + \tau^2 E \left \{ \frac{(\Delta^\star + \tilde \delta )^\top L^\top A L ( \Delta^\star + \tilde \delta)}{\xi^2} \right \} \\
        &- 2 \tau E \left \{ \frac{(\Delta^\star+\tilde \delta) L^\top A P_\gamma Q_I^{-1}\Delta^\star}{\xi} \right \} \\
        =& R( \hat \gamma_I, \gamma_0) + \tau^2 E \left \{ 1/\xi \right \} 
        - 2 \tau E \left \{ \frac{(\Delta^\star+\tilde \delta) L^\top A P_\gamma Q_I^{-1}\Delta^\star}{\xi} \right \} 
    \end{align*}
   Then we can re-write the final term as
    $$
    \underbrace{\frac{(\Delta^\star + \tilde \delta)}{(\Delta^\star + \tilde \delta)^\top B (\Delta^\star + \tilde \delta)}}_{g(\Delta^\star + \tilde \delta)} \cdot \underbrace{(L^\top A P_{\gamma} Q_I^{-1})}_{K} \Delta^\star := g(\Delta^\star + \tilde \delta) K \Delta^\star.
    $$
    Let~$\phi_{\Delta}(z)$ denote the density for $\Delta^\star$. Then 
    \begin{align*}
    &E \left \{ g(\Delta^\star + \tilde \delta)^\top K \Delta^\star \right \} \\
    = 
    &\int g(\Delta^\star + \tilde \delta) K \bfz \phi_{\Delta^\star} (\bfz) d\bfz \\
    =&- \int g(\Delta^\star + \tilde \delta) K W^\star d \phi_{\Delta^\star} (\bfz) \\ 
    =&\int \text{tr} \left \{ \frac{d}{d\bfz} g(\Delta + \tilde \delta) K W \right \} \phi_{\Delta} (\bfz) d\bfz \\
    =& \int \text{tr} \left \{
    \frac{K W^\star}{(\bfz + \tilde \delta)^\top B (\bfz + \tilde \delta)} - 
    \frac{2 B (\bfz + \tilde \delta)
    (\bfz + \tilde \delta)^\top K W^\star
    }{\{(\bfz + \tilde \delta)^\top B (\bfz + \tilde \delta)\}^2}
    \right \} \phi_{\Delta} (\bfz) d\bfz \\
    =&E \text{tr} \left \{
    \frac{K W^\star}{\xi} \right\}
    -2 E tr \left\{
    \frac{B (\Delta^\star + \tilde \delta)
    (\Delta^\star + \tilde \delta)^\top K W^\star
    }{\xi^2}
    \right \}.
    \end{align*}
where the fourth equality follows integration by parts.  Then studying the first term, we have
\begin{align*}
tr(K W^\star) &= tr \left \{L^\top A P_\gamma Q_I^{-1} W^\star \right \}\\
&= tr \left \{  \underbrace{P_\gamma Q_I^{-1} W^\star Q_I^{-1} H^\top P_{\theta}^\top}_{\Sigma_{\theta, \gamma}^\top} \Sigma_\theta^{-1}\Sigma_{\theta, \gamma} A \right \}\\
&= tr \left \{ \Sigma_\theta^{-1/2} \Sigma_{\theta, \gamma} A \Sigma_{\theta, \gamma} \Sigma_{\theta}^{-1/2} \right \} \\
&= tr \left \{ J \right \},
\end{align*}
where~$J := \Sigma_\theta^{-1/2} \Sigma_{\theta, \gamma} A \Sigma_{\gamma, \theta} \Sigma_{\theta}^{-1/2}$ by definition.  
We show the equality of the underbraces at the end of the proof.
We next try and write the second term using $J$.
\begin{align*}
    &tr \left \{ B (\Delta^\star + \tilde \delta) (\Delta^\star + \tilde \delta)^\top KW^\star \right \} \\
    = &\left \{ (\Delta^\star + \tilde \delta)^\top KW^\star B (\Delta^\star + \tilde \delta) \right \} \\
    = &\left \{ (\Delta^\star + \tilde \delta)^\top L^\top A P_\gamma Q_I^{-1} W^\star L^\top A L (\Delta^\star + \tilde \delta) \right \} \\
    = &\left \{ (\Delta^\star + \tilde \delta)^\top L^\top A \underbrace{P_\gamma Q_I^{-1} W^\star Q_I^{-1} H^\top P_{\theta}^\top}_{= \Sigma_{\gamma, \theta}} \Sigma_\theta^{-1} \Sigma_{\theta, \gamma} A L (\Delta^\star + \tilde \delta) \right \} \\
     = &\left \{ (\Delta^\star + \tilde \delta)^\top L^\top A \Sigma_{\gamma, \theta} \Sigma_\theta^{-1} \Sigma_{\theta, \gamma} A L (\Delta^\star + \tilde \delta) \right \} \\
      = &\left \{ (\Delta^\star + \tilde \delta)^\top \underbrace{Q_I^{-1} H^\top P_\theta^\top \Sigma_\theta^{-1} \Sigma_{\theta, \gamma}}_{L^\top} A \Sigma_{\gamma, \theta} \Sigma_\theta^{-1} \Sigma_{\theta, \gamma} A \underbrace{\Sigma_{\gamma, \theta} \Sigma_{\theta}^{-1} P_{\theta} H Q_I^{-1}}_{L} (\Delta^\star + \tilde \delta) \right \} \\
       = &\left \{ (\Delta^\star + \tilde \delta)^\top \underbrace{Q_I^{-1} H^\top P_\theta^\top \Sigma_\theta^{-1/2}}_{B_1^\top} \underbrace{\Sigma_{\theta}^{-1/2} \Sigma_{\theta, \gamma} A \Sigma_{\gamma, \theta} \Sigma_\theta^{-1/2}}_{J} \underbrace{\Sigma_{\theta}^{-1/2} \Sigma_{\theta, \gamma} A \Sigma_{\gamma, \theta} \Sigma_{\theta}^{-1/2}}_{J} \underbrace{\Sigma_{\theta}^{-1/2} P_{\theta} H Q_I^{-1}}_{B_1} (\Delta^\star + \tilde \delta) \right \} \\
        = &\left \{ (\Delta^\star + \tilde \delta)^\top \underbrace{B_1^\top J^{1/2}}_{\tilde B^\top} J \underbrace{J^{1/2} B_1}_{\tilde B} (\Delta^\star + \tilde \delta) \right \} \\
        = &\left \{ (\Delta^\star + \tilde \delta)^\top \tilde B^\top J \tilde B (\Delta^\star + \tilde \delta) \right \}.
\end{align*}
Then
\begin{align*}
\tilde B^\top \tilde B &= B_1^\top J B_1 \top \\
&= Q^{-1}_I H^\top P^\top_{\theta} \Sigma_\theta^{-1/2} \Sigma_\theta^{-1/2} \Sigma_{\theta, \gamma} A \Sigma_{\gamma,\theta} \Sigma_{\theta}^{-1/2} \Sigma_\theta^{-1/2} P_{\theta} H Q_I^{-1} \\
&= \underbrace{Q_I^{-1} H^\top P^\top_{\theta} \Sigma_{\theta}^{-1} \Sigma_{\theta,\gamma}}_{L^\top} A \underbrace{\Sigma_{\gamma, \theta} \Sigma_{\theta}^{-1} P_\theta H Q_I^{-1}}_{L} \\
&= L^\top A L = B.
\end{align*}
Recall the Rayleigh quotient inequality, $Y^\top J Y \leq \|J\| Y^\top Y$ where $\| J \|$ is the spectral norm.  In our setting~$Y = \tilde B (\Delta^\star + \tilde \delta)$ implies the following inequality
$$
E \left( \frac{tr \left \{ B (\Delta^\star + \tilde \delta) (\Delta^\star + \tilde \delta)^\top KW^\star \right \}}{\xi^2} \right) \leq E \left ( \frac{\|J\|}{\xi} \right )
$$
which then implies the following inequality
$$
E \left \{ g(\Delta^\star + \tilde \delta)^\top K \Delta^\star \right \} \geq E \left \{ \frac{tr(J) - 2 \| J \|}{\xi} \right\} 
$$
From this inequality and the other terms above, we have
$$
R( \hat \gamma_{JS}, \gamma^\star) <
R( \hat \gamma_I, \gamma^\star) - \tau
\times E \left ( 
\frac{tr(J) - 2\| J \| - \tau}{\xi}
\right) 
\leq 
R( \hat \gamma_I, \gamma^\star) - \tau
\times \left ( 
\frac{tr(J) - 2\| J \| - \tau}{E(\xi)}
\right) 
$$
where the second line holds by Jensen's inequality and $d := tr(J)/\|J\| > 2$ for any $0 < \tau < 2 \{ tr(J) - \|J\| \}$. Optimal choice of $\tau$ is $\tau^\star = tr(J) - 2 \| J \|$.

\noindent {\bf Confirming equality.} Recall~$M_{IE} = Q_{I, \theta}^{-1} Q_{E,\theta}$. Define
$$
    W^\star  = \left( \begin{array}{cc}
         \frac{1}{c} M_{IE} W_{\theta, E} M_{IE}^\top + W_{\theta I} &  W_{\theta, \gamma, I} \\
          W_{\theta, \gamma, I} & W_{\gamma, I}
    \end{array} \right)
$$
Then~$Q_{I}^{-1} W^\star Q_{I}^{-1}$ is given by 
\begin{align*}
    &\left( \begin{array}{cc}
         Q_{\theta,I}^{-1} & 0 \\
         0 & Q_{\gamma,I}^{-1}
    \end{array} \right)
\left( \begin{array}{cc}
         \frac{1}{c} M_{IE} W_{\theta, E} M_{IE}^\top + W_{\theta I} & W_{\theta, \gamma, I} \\
          W_{\theta, \gamma, I} & W_{\gamma, I}
    \end{array} \right) 
     \left( \begin{array}{cc}
         Q_{\theta,I}^{-1} & 0 \\
         0 & Q_{\gamma,I}^{-1}
    \end{array} \right)
    \\
    =&
    \left( \begin{array}{cc}
         \frac{1}{c} Q_{\theta,E}^{-1} W_{\theta, E} M_{IE}^\top + Q_{\theta, I}^{-1} W_{\theta I} & Q_{\theta, I} W_{\theta, \gamma, I} \\
          Q_{\gamma,I}^{-1} W_{\theta, \gamma, I} & Q_{\gamma, I}^{-1} W_{\gamma, I}
    \end{array} \right)  
    \left( \begin{array}{cc}
         Q_{\theta,I}^{-1} & 0 \\
         0 & Q_{\gamma,I}^{-1}
    \end{array} \right) \\
    =&
    \left( \begin{array}{cc}
         \frac{1}{c} Q_{\theta,E}^{-1} W_{\theta, E} Q_{\theta,E}^{-1} + Q_{\theta, I}^{-1} W_{\theta I} Q_{\theta,I}^{-1} &  Q_{\theta, I} W_{\theta, \gamma, I} Q_{\gamma,I}^{-1} \\
          Q_{\gamma,I}^{-1} W_{\theta, \gamma, I} Q_{\theta,I}^{-1} & Q_{\gamma, I}^{-1} W_{\gamma, I} Q_{\gamma,I}^{-1}
    \end{array} \right)  \\
    =&
    \left( \begin{array}{cc}
         Q_{\theta, I}^{-1} W_{\theta I} Q_{\theta,I}^{-1} &   Q_{\theta, I} W_{\theta, \gamma, I} Q_{\gamma,I}^{-1} \\
          Q_{\gamma,I}^{-1} W_{\theta, \gamma, I} Q_{\theta,I}^{-1} & Q_{\gamma, I}^{-1} W_{\gamma, I} Q_{\gamma,I}^{-1}
    \end{array} \right)   +
    \left( \begin{array}{cc}
         \frac{1}{c} Q_{\theta,E}^{-1} W_{\theta, E} Q_{\theta,E}^{-1} 
         & 0 \\
         0 & 0
    \end{array} \right)  \\
    &= \left( \begin{array}{cc}
         \Sigma_{\theta,I} &   \Sigma_{\theta, \gamma, I} \\
          \Sigma_{\gamma,\theta, I}  & \Sigma_{\gamma, I}
    \end{array} \right)   +
    \left( \begin{array}{cc}
         \frac{1}{c} \Sigma_{\theta,E} 
         & 0 \\
         0 & 0
    \end{array} \right) 
\end{align*}
Which implies, that 
$$
P_{\gamma} H Q_{I}^{-1} W^\star Q_{I}^{-1} H^\top P_{\theta}^\top
= \Sigma_{\gamma, \theta,I}
$$
\end{proof}

\section{Code Availability}
The code for this project is available on GitHub at \url{https://github.com/wdempsey/JS_Mestimators}.

\bibliographystyle{unsrt}
\bibliography{paper-ref}

\end{document}